\documentclass[a4paper,USenglish]{lipics-v2018}
\usepackage{xspace}
\usepackage{mathtools}

\graphicspath{{./Figures/}}

\title{Computing the Chromatic Number Using Graph Decompositions via Matrix Rank}
\titlerunning{Chromatic Number via Matrix Rank}
\authorrunning{Bart M.\,P. Jansen and Jesper Nederlof}

\author{Bart M.\,P. Jansen}{Eindhoven University of Technology, Eindhoven, The Netherlands}{b.m.p.jansen@tue.nl}{}{NWO Veni grant ``Frontiers in Parameterized Preprocessing'' and NWO Gravitation grant ``Networks''}
\author{Jesper Nederlof}{Eindhoven University of Technology, Eindhoven, The Netherlands}{j.nederlof@tue.nl}{}{ NWO Veni grant  ``Reducing small instances of complex tasks to large instances of simple ones'' and NWO Gravitation grant ``Networks''}

\subjclass{%
; 
}	

\Copyright{Bart M.P. Jansen and Jesper Nederlof}


\EventEditors{Yossi Azar, Hannah Bast, and Grzegorz Herman}
\EventNoEds{3}
\EventLongTitle{26th Annual European Symposium on Algorithms (ESA 2018)}
\EventShortTitle{ESA 2018}
\EventAcronym{ESA}
\EventYear{2018}
\EventDate{August 20--22, 2018}
\EventLocation{Helsinki, Finland}
\EventLogo{}
\SeriesVolume{112}
\ArticleNo{47} 

\nolinenumbers 
\hideLIPIcs  

\theoremstyle{plain}
\newtheorem{claim}[theorem]{Claim}
\newtheorem{proposition}[theorem]{Proposition}
\newtheorem{observation}[theorem]{Observation}
\newtheorem*{seth}{Strong Exponential Time Hypothesis}
\newcommand{\qSAT}{\textsc{$q$-SAT}\xspace}
\newcommand{\Oh}{\ensuremath{\mathcal{O}}\xspace}
\newcommand{\qColoring}{\textsc{$q$-Coloring}\xspace}

\newcommand{\tw}{\ensuremath{\mathrm{tw}}\xspace}
\newcommand{\cutw}{\ensuremath{\mathrm{ctw}}\xspace}
\newcommand{\pw}{\ensuremath{\mathrm{pw}}\xspace}

\newcommand{\bigmid}{\,\big\vert\,}
\newcommand{\true}{\textsc{true}\xspace}
\newcommand{\false}{\textsc{false}\xspace}

\newenvironment{claimproof}[1][\proofname]{\begin{proof}[#1]\renewcommand{\qedsymbol}{\claimqed}}{\end{proof}\renewcommand{\qedsymbol}{\plainqed}}

\let\plainqed\qedsymbol

\begin{document}

\maketitle
\begin{abstract}
Computing the smallest number~$q$ such that the vertices of a given graph can be properly~$q$-colored is one of the oldest and most fundamental problems in combinatorial optimization. The \textsc{$q$-Coloring} problem has been studied intensively using the framework of parameterized algorithmics, resulting in a very good understanding of the best-possible algorithms for several parameterizations based on the structure of the graph. For example, algorithms are known to solve the problem on graphs of treewidth~$\tw$ in time~$\mathcal{O}^*(q^\tw)$, while a running time of~$\mathcal{O}^*((q-\varepsilon)^\tw)$ is impossible assuming the Strong Exponential Time Hypothesis (SETH). While there is an abundance of work for parameterizations based on decompositions of the graph by \emph{vertex separators}, almost nothing is known about parameterizations based on \emph{edge separators}. We fill this gap by studying \textsc{$q$-Coloring} parameterized by cutwidth, and parameterized by pathwidth in bounded-degree graphs. Our research uncovers interesting new ways to exploit small edge separators.

We present two algorithms for \textsc{$q$-Coloring} parameterized by cutwidth~$\cutw$: a deterministic one that runs in time~$\mathcal{O}^*(2^{\omega \cdot \cutw})$, where~$\omega$ is the matrix multiplication constant, and a randomized one with runtime~$\mathcal{O}^*(2^{\cutw})$. In sharp contrast to earlier work, the running time is \emph{independent} of~$q$. The dependence on cutwidth is optimal: we prove that even \textsc{3-Coloring} cannot be solved in~$\mathcal{O}^*((2-\varepsilon)^{\cutw})$ time assuming SETH. Our algorithms rely on a new rank bound for a matrix that describes compatible colorings. Combined with a simple communication protocol for evaluating a product of two polynomials, this also yields an~$\mathcal{O}^*((\lfloor d/2\rfloor+1)^{\pw})$ time randomized algorithm for \textsc{$q$-Coloring} on graphs of pathwidth~$\pw$ and maximum degree~$d$. Such a runtime was first obtained by Bj\"orklund, but only for graphs with few proper colorings. We also prove that this result is optimal in the sense that no~$\mathcal{O}^*((\lfloor d/2\rfloor+1-\varepsilon)^{\pw})$-time algorithm exists assuming SETH.
\end{abstract}

\section{Introduction}
Graph coloring is one of the most fundamental combinatorial problems, studied already in the 1850s. Countless papers (cf.~\cite{Pardalos1998}) and several monographs~\cite{JensenT1995,Johnson2008,Lewis15} have been devoted to its combinatorial and algorithmic investigation. Since the graph coloring problem is NP-complete even in restricted settings such as planar graphs~\cite{GareyJS76}, considerable effort has been invested in finding polynomial-time approximation algorithms and exact algorithms that beat brute-force search~\cite{BjorklundH08,BjorklundHK09}. 

A systematic study of which characteristics of inputs govern the complexity of the graph coloring problem has been undertaken using the framework of parameterized algorithmics. The aim in this framework is to obtain algorithms whose running time is of the form~$f(k) \cdot n^{\Oh(1)}$, where~$k$ is a parameter that measures the complexity of the instance and is independent of the number of vertices~$n$ in the input graph. Over the past decade, numerous parameters have been employed that quantify the structure of the underlying graph. In several settings, algorithms have been obtained that are \emph{optimal} under the Strong Exponential Time Hypothesis (SETH)~\cite{ImpagliazzoP01,ImpagliazzoPZ01}. For example, it has long been known (cf.~\cite[Theorem 7.9]{CyganFKLMPPS15},\cite{TelleP97}) that testing $q$-colorability on a graph that is provided together with a tree decomposition of width~$k$ can be done in time~$\Oh(q^k \cdot k^{\Oh(1)} \cdot n)$.  Lokshtanov, Marx, and Saurabh~\cite{LokshtanovMS11} proved a matching lower bound: an algorithm running in time~$(q-\varepsilon)^k \cdot n^{\Oh(1)}$ for any~$\varepsilon > 0$ and integer~$q \geq 3$ would contradict SETH. Results are also known for graph coloring parameterized by the vertex cover number~\cite{JaffkeJ17}, pathwidth and the feedback vertex number~\cite{LokshtanovMS11}, cliquewidth~\cite{FominGLS10,GolovachL0Z18,KoblerR03}, twin-cover~\cite{Ganian11}, modular-width~\cite{GajarskyLO13}, and split-matching width~\cite{SaetherT16}. (See~\cite[Fig.~1]{FellowsJR13} for relations between these parameters.)

A survey of these algorithmic results for graph coloring results in the following picture of the complexity landscape: For graph parameters that are defined in terms of the width of decompositions by vertex separators (pathwidth, treewidth, vertex cover number, etc.), one can typically obtain a running time\footnote{We use~$\Oh^*(f(k))$ as a shorthand for~$f(k) \cdot n^{\Oh(1)}$.} of~$\Oh^*(q^k)$ to test whether a graph that is given together with a decomposition of width~$k$ is $q$-colorable, but assuming (S)ETH there is no algorithm with running time~$\Oh^*(c^k)$ for any constant~$c$ independent of~$q$~\cite[Theorem 11]{JaffkeJ17}.

The complexity of graph coloring parameterized by width measures based on vertex separators is therefore well-understood by now. However, only little attention has been paid to graph decompositions whose width is measured in terms of the number of \emph{edges} in a separator. There is intriguing evidence that separators consisting of few edges (or, equivalently, consisting of a bounded number of bounded-degree vertices) can be algorithmically exploited in nontrivial ways when solving \textsc{$q$-Coloring}. In 2016, Bj\"orklund~\cite{Bjorklund16} presented a fascinating algebraic algorithm that decides $q$-colorability using an algorithmic variation on the Alon-Tarsi theorem~\cite{AlonT92}. Given a graph~$G$ of maximum degree~$d$, a path decomposition of width~$k$, and integers~$q$ and~$s$, his algorithm runs in time~$(\lfloor d / 2 \rfloor + 1)^k n^{\Oh(1)} \cdot s$. If the graph is not $q$-colorable it always outputs \textsc{no}. If the graph has at most~$s$ proper $q$-colorings, then it outputs \textsc{yes} with constant probability. Hence when~$q \geq (\lfloor d / 2 \rfloor + 1)$ and~$s$ is small, it improves over the standard~$\Oh^*(q^k)$-time dynamic program by exploiting the bounded-degree vertex separators encoded in the path decomposition. However, the dependence of the running time on the number of proper $q$-colorings in the graph is very undesirable, as that number may be exponentially large in~$n$.

Bj\"orklund's algorithm hints at the fact that graph decompositions whose width is governed by the number of \emph{edges} in a separator may yield an algorithmic advantage over existing approaches. In this work, we therefore perform a deeper investigation of how decompositions by small edge separators can be exploited when solving \textsc{$q$-Coloring}. By leveraging a new rank upper bound for a matrix that describes the compatibility of colorings of subgraphs on two sides of a small edge separator, we obtain a number of novel algorithmic results. In particular, we show how to eliminate dependence on the number~$s$ of proper colorings.

\subparagraph{Our results} 
We present efficient algorithms for \textsc{$q$-Coloring} parameterized by the width of various types of graph decompositions by small edge separators. Our first results are phrased in terms of the graph parameter \emph{cutwidth}. A decomposition in this case corresponds to a linear ordering of the vertices; the cutwidth of this ordering is given by the maximum number of edges that connect a vertex in a prefix of the ordering to a vertex in the complement (see Section~\ref{sec:definitions} for formal definitions). Cutwidth is one of the classic graph layout parameters (cf.~\cite{DiazPS02}). It takes larger values than treewidth~\cite{KorachS93}, and has been the subject of frequent study~\cite{GiannopoulouPRT16,ThilikosSB05,ThilikosSB05a}.

Informally speaking, we prove that interactions of partial solutions on low-cutwidth graphs are much simpler than interactions of partial solutions on low-pathwidth graphs. The rank-based approach developed in earlier work~\cite{BodlaenderCKN15,CyganKN13,FominLPS16} can be used by setting up matrices whose rank determines the complexity of these interactions in low-cutwidth graphs. These are different from the matrices associated to partial solutions in low-pathwidth graphs, and admit better rank bounds. This is exploited by two different algorithms: a deterministic algorithm that employs fast matrix multiplication and therefore has the matrix-multiplication constant~$\omega$ in its running time, and a faster randomized Monte Carlo algorithm.

\begin{theorem}\label{thm:cw}
There is a deterministic algorithm that, for any~$q$, solves \textsc{$q$-Coloring} on a graph~$G$ with a given linear layout of cutwidth~$\cutw$ in~$\Oh^*(2^{\omega\cdot \cutw})$ time, where $\omega\leq2.373$ is the matrix multiplication constant.
\end{theorem}

\begin{theorem} \label{thm:coloring:randomized}
There is a randomized Monte Carlo algorithm that, for any~$q$, solves \textsc{$q$-Coloring} on a graph~$G$ with a given linear layout of cutwidth~$\cutw$ in~$\Oh^*(2^\cutw)$ time.
\end{theorem}

These results show a striking difference between cutwidth and parameterizations based on vertex separators such as treewidth and vertex cover number: we obtain single-exponential running times where the base of the exponent is \emph{independent} of the number of colors~$q$, which (assuming ETH) is impossible even parameterized  by vertex cover~\cite{JaffkeJ17}. The assumption that a decomposition is given in the input is standard in this line of research~\cite{BodlaenderCKN15,CyganNPPRW11,CyganKN13,FominLPS16} and decouples the complexity of \emph{finding} a decomposition from that of \emph{exploiting} a decomposition.

The ideas underlying Theorems~\ref{thm:cw} and~\ref{thm:coloring:randomized} can also be used to eliminate the dependence on the number of proper colorings from Bj\"orklund's algorithm. We prove the following theorem:

\begin{theorem} \label{thm:alg:degree}
There is a randomized Monte Carlo algorithm that, for any~$q$, solves \textsc{$q$-Coloring} on a graph~$G$ with maximum degree $d$ and given path decomposition of width $\pw$ in $\Oh^*((\lfloor d/2\rfloor+1)^\pw)$ time.
\end{theorem}

Our approach uses the first step of the proof of the Alon-Tarsi theorem (i.e.~rewrite the problem into evaluating the graph polynomial) and also relates colorability to certain orientations, but deviates from the previous algorithm otherwise: to evaluate the appropriate graph polynomial we extend a fairly simple communication-efficient protocol to evaluate a product of two polynomials.

We also prove that the randomized algorithms of Theorem~\ref{thm:coloring:randomized} and Theorem~\ref{thm:alg:degree} are conditionally \emph{optimal}, even when restricted to special cases:
\begin{theorem}\label{thm:3collow}
	Assuming SETH, there is no~$\varepsilon > 0$ such that \textsc{$3$-Coloring} on a planar graph~$G$ given along with a linear layout of cutwidth~$\cutw$ can be solved in time~$\Oh^*((2-\varepsilon)^{\cutw})$.
\end{theorem}

\begin{theorem}\label{thm:degreecollow}
Let~$d \geq 5$ be an odd integer and let~$q_d := \lfloor d /2 \rfloor + 1$. Assuming SETH, there is no~$\varepsilon > 0$ such that \textsc{$q_d$-Coloring} on a graph of maximum degree~$d$ given along with a path decomposition of pathwidth~$\pw$ can be solved in time~$\Oh^*((\lfloor d/2 \rfloor + 1 - \varepsilon)^{\pw})$.
\end{theorem}

These results are obtained by building on the techniques of Lokshtanov et al.~\cite{LokshtanovMS11} that propagate `partial assignments' throughout graphs of small cutwidth or pathwidth.

\subsection*{Organization}
In Section~\ref{sec:definitions} we provide preliminaries.
In Section~\ref{sect:upp} we present algorithms for graph coloring, proving Theorems~\ref{thm:cw},~\ref{thm:coloring:randomized}, and~\ref{thm:alg:degree}.
In Section~\ref{sect:lowmain} we present reductions that show that our randomized algorithms cannot be improved significantly, assuming SETH, proving Theorems~\ref{thm:3collow} and~\ref{thm:degreecollow}.
Finally, we provide some conclusions in Section~\ref{sec:conc}. Due to space restrictions, several proofs had to be moved to the appendix.

\section{Preliminaries}\label{sec:definitions} \label{sec:preliminaries}
We use~$\mathbb{N}$ to denote the natural numbers, including~$0$. For a positive integer~$n$ and a set~$X$ we use~$\binom{X}{n}$ to denote the collection of all subsets of~$X$ of size~$n$. The \emph{power set} of~$X$ is denoted~$2^{X}$. The set~$\{1, \ldots, n\}$ is abbreviated as~$[n]$. The~$\Oh^*$ notation suppresses polynomial factors in the input size~$n$, such that~$\Oh^*(f(k))$ is shorthand for~$\Oh(f(k) n^{\Oh(1)})$. All our logarithms have base two. For sets $S,T$ we denote by $S^T$ the set of vectors indexed by elements of $T$ whose entries are from $S$. If $T=[n]$, we use $S^{n}$ instead of $S^{[n]}$.

We consider finite, simple, and undirected graphs~$G$, consisting of a vertex set~$V(G)$ and edge set~$E(G) \subseteq \binom{V(G)}{2}$. The neighbors of a vertex~$v$ in~$G$ are denoted~$N_G(v)$. The closed neighborhood of~$v$ is~$N_G[v] := N_G(v) \cup \{v\}$. The degree $d(v)$ equals $|N_G(v)|$ and if $X \subseteq E(G)$, then $d_X(v)$ denotes the number of edges of~$X$ incident to $v$. This notation is extended to $d^-(v),d^+(v),d^-_X(v),d^+_X(v)$ for directed graphs in the natural way (e.g.~$d^+_X(v)$ denotes the number of $w$ such that $(v,w) \in X$). For a vertex set~$S \subseteq V(G)$ the open neighborhood is~$N_G(S) := \bigcup_{v \in S} N_G(v) \setminus S$ and the closed neighborhood is~$N_G[S] := N_G(S) \cup S$, while~$G[S]$ denotes the graph induced by $S$.

A $q$-coloring of a graph~$G$ is a function~$f \colon V(G) \to [q]$. A coloring is \emph{proper} if~$f(u) \neq f(v)$ for all edges~$\{u,v\} \in E(G)$. For a fixed integer~$q$, the \textsc{$q$-Coloring} problem asks whether a given graph~$G$ has a proper $q$-coloring. The \qSAT problem asks whether a given Boolean formula, in conjunctive normal form with clauses of size at most~$q$, has a satisfying assignment.

\begin{seth}[\cite{ImpagliazzoP01,ImpagliazzoPZ01}]
For every $\varepsilon > 0$, there is a constant~$q$ such that \qSAT on $n$ variables cannot be solved in time $\Oh^*((2-\varepsilon)^n)$.
\end{seth}

\subparagraph{Cutwidth} 
For an $n$-vertex graph~$G$, a \emph{linear layout} of~$G$ is a linear ordering of its vertex set, given by a bijection~$\pi \colon V(G) \to [n]$. The \emph{cutwidth} of~$G$ with respect to the layout~$\pi$ is:
$$\cutw_\pi(G) = \max_{1 \leq i < n} \bigl | \bigl \{ \{u,v\} \in E(G) \bigmid \pi(u) \leq i \wedge \pi(v) > i \bigr\} \bigr |,$$
and the cutwidth~$\cutw(G)$ of a graph~$G$ is the minimum cutwidth attained by any linear layout. It is well-known (cf.~\cite{Bodlaender98}) that~$\cutw(G) \geq \pw(G) \geq \tw(G)$, where the latter denote the pathwidth and treewidth of~$G$, respectively. An intuitive way to think about cutwidth is to consider the vertices as being placed on a horizontal line in the order dictated by the layout~$\pi$, with edges drawn as $x$-monotone curves. For any position~$i$ we consider the gap between vertex~$\pi^{-1}(i)$ and~$\pi^{-1}(i+1)$, and count the edges that \emph{cross} the gap by having one endpoint at position at most~$i$ and the other at position after~$i$. The cutwidth of a layout is the maximum number of edges crossing any single gap.

\newcommand{\pathdecomp}{\ensuremath{P}}
\subparagraph*{Pathwidth and path decompositions} A \emph{path decomposition} of a graph~$G$ is a path~$\pathdecomp$ in which each node~$x$ has an associated set of vertices~$B_x\subseteq V(G)$ (called a \emph{bag}) such that~$\bigcup_{x \in V(P)} B_x=V(G)$ and the following properties hold:
\begin{enumerate}
	\item For each edge~$\{u,v\}\in E(G)$ there is a node~$x$ in~$\pathdecomp$ such that~$u,v\in B_x$.
	\item If~$v\in B_x\cap B_y$ then~$v\in B_z$ for all nodes~$z$ on the (unique) path from~$x$ to~$y$ in~$\pathdecomp$.
\end{enumerate}
The \emph{width} of~$\pathdecomp$ is the size of the largest bag minus one, and the pathwidth of a graph~$G$ is the minimum width over all possible path decompositions of~$G$. Since our focus here is on dynamic programming over a path decomposition we only mention in passing that the related notion of treewidth can be defined in the same way, except for letting the nodes of the decomposition form a tree instead of a path.

It is common for the presentation of dynamic-programming algorithms to use path- and tree decompositions that are normalized in order to make the description easier to follow. For an overview of tree decompositions and dynamic programming on tree decompositions see e.g.~\cite{BodlaenderK08}. Following~\cite{CyganNPPRW11} we use the following path decompositions:

\begin{definition}[Nice Path Decomposition] \label{def:nicepathdecomp}
	A \emph{nice path decomposition} is a path decomposition where the underlying path of nodes is ordered from left to right (the predecessor of any node is its left neighbor) and in which each bag is of one of the following types:
	\begin{itemize}
		\item \textbf{First (leftmost) bag}: the bag associated with the leftmost node~$x$ is empty,~$B_x=\emptyset$.
		\item \textbf{Introduce vertex bag}: an internal node~$x$ of $\pathdecomp$ with predecessor~$y$ such that~$B_x = B_y \cup \{v\}$ for some $v \notin B_y$. 
		This bag is said to \emph{introduce} $v$.
		\item \textbf{Introduce edge bag}: an internal node~$x$ of $\pathdecomp$ labeled 
		with an edge $\{u,v\} \in E(G)$ with one predecessor~$y$ for which $u,v \in B_x = B_y$. 
		This bag is said to \emph{introduce} $uv$.
		\item \textbf{Forget bag}: an internal node~$x$ of $\pathdecomp$ with one predecessor~$y$ for which $B_x = B_y \setminus \{v\}$ for some $v \in B_y$. This bag is said to \emph{forget} $v$.
		\item \textbf{Last (rightmost) bag}: the bag associated with the rightmost node~$x$ is empty,~$B_x=\emptyset$.
	\end{itemize}
\end{definition}

It is easy to verify that any given path decomposition of pathwidth~$\pw$ can be transformed in time~$|V(G)| \cdot \pw^{\Oh(1)}$ into a nice path decomposition without increasing the width.
Let $B_1,\ldots,B_\ell$ be a nice path decomposition of $G$. We say $B_i$ is \emph{before} $B_j$ if $i\leq j$. We denote $V_i=\bigcup_{j=1}^i B_i$ and let $E_i$ denote the set of edges introduced in bags before $i$.

\section{Upper bounds for Graph Coloring}\label{sect:upp}
\newcommand{\cS}{\ensuremath{\mathcal{S}}}
\newcommand{\poly}{\ensuremath{\mathrm{poly}}}
In this section we outline algorithms for \qColoring that run efficiently when given a graph and either a small-cutwidth layout or a good path decomposition on graphs with small maximum degree. We assume the input graph has no isolated vertices, as they are clearly irrelevant. We start by using the `rank-based approach' as proposed in~\cite{BodlaenderCKN15} to obtain deterministic algorithms, and afterward give a randomized algorithm with substantial speedup. In both approaches the idea is to employ dynamic programming to accumulate needed information about the existence of partial solutions, but use linear-algebraic methods to compress this information. Let us remark in passing that our approaches are robust in the sense that they directly extend to generalizations such as \textsc{$q$-List Coloring} in which for every vertex a set of allowed colors is given.\footnote{In the deterministic approach we simply avoid partial solutions not satisfying these constraints, and in the randomized approach we assign sufficiently large weight to disallowed (vertex,color) combinations.}

A key quantity that determines the amount of information needed after compression in general is the rank of a \emph{partial solutions matrix}. This matrix has its rows and columns indexed by partial solutions (which could be defined in various ways) and an entry is $1$ (or more generally, non-zero) if the two partial solutions combine to a solution. Previously, this method proved to be highly useful for connectivity problems parameterized by treewidth~\cite{BodlaenderCKN15}. For \qColoring parameterized by treewidth, partial solutions can naturally be defined as partial proper colorings of a subgraph whose boundary is formed by some vertex separator. Two partial colorings combine to a proper complete coloring if and only if the two partial colorings agree on the coloring of the separator. Unfortunately, the rank-based approach is not useful here as the  partial solution matrices arising have large rank, as witnessed by induced identity submatrices of dimensions $q^{\tw}$. Indeed, the lower bound under SETH by Lokshtanov, Marx, and Saurabh~\cite{LokshtanovMS11} shows that no algorithm can solve the problem much faster than $\Oh^*(q^{\pw})$, where $\pw$ denotes the pathwidth of the input graph.

Still, this does not exclude much faster running times parameterized by \emph{cutwidth}. In our application of the rank-based approach for \qColoring of a graph with a given linear layout of cutwidth~$\cutw$, the partial solutions are $q$-colorings of the first $i$ and last $n-i$ vertices in the linear order, and clearly only the colors assigned to vertices incident to the edges going over the cut are relevant. If we let $X=X_i,Y=Y_i$ denote the endpoints of these edges occurring respectively not after and after $i$, and let $H=H_i$ denote the bipartite graph induced by the cut and these edges, we are set to study the rank of the following partial solutions matrix indexed by $x\in [q]^{X}$ and $y \in [q]^{Y}$:
\[
	M_H[x,y] = \begin{cases}
		1, & \text{if $x \cup y$ is a proper $q$-coloring of $H$},\\
		0, & \text{if otherwise}.
	\end{cases}
\]
Here and below, we slightly abuse notation by viewing elements of $V^I$ (i.e.~vectors with values in $V$ that are indexed by $I$) as sets of pairs in $I \times V$; that is, if $x \in V^I$ we also use $x$ to denote the set $\{(i,x_i)\}_{i \in I}$. With this notation in mind, note that $x \cup y$ above can be interpreted as an element of $[q]^{X \cup Y}$ in the natural way as $X$ and $Y$ are disjoint.
As the rank of~$M_H$ is generally high\footnote{For example, if $H$ is a single edge $M_H$ is the complement of an identity matrix of dimensions $q \times q$.} and depends on~$q$, we instead focus on the matrix $M'_H$ defined by
\begin{equation}\label{eq:M'def}
	M'_H[x,y] = \prod_{(v,w) \in E(H)} (x_v - y_w),
\end{equation}
where all edges are directed from $X$ to $Y$ in $E(H)$. The crux is that the support (e.g.~the set of non-zero entries) of $M'_H$ equals the support of $M_H$:
\begin{lemma} \label{lem:zero}
	We have~$M'_H[x,y] \neq 0$ if and only if $x \cup y$ is a proper $q$-coloring of $H$.
\end{lemma}
\begin{proof}
	If~$x_v = y_w$ for some $(v,w) \in E(H)$ then the term~$(x_v - y_w)$ is zero, implying the entire product on the right hand-side of~\eqref{eq:M'def} is zero. If~$x$ and~$y$ differ at every coordinate, then~$M'_H[x,y]$ is a product of nonzero terms, and therefore non-zero itself.
\end{proof}


In Sections~\ref{subsec:uppcoldet}--\ref{subsec:uppcolrand} this property will allow us to work with~$M'_H$ instead of~$M_H$, when combined with the Isolation Lemma or Gaussian-elimination approach; similarly as in previous work~\cite{BodlaenderCKN15,CyganKN13,CyganNPPRW11}.\footnote{In the deterministic setting, the observation that one can work with a matrix different from a partial solution matrix but with the same support as the partial solution matrix was already used by Fomin et al.~\cite{FominLPS16} in combination with a matrix factorization by Lov\'asz~\cite{lovasz1977flats}.}

\subsection{A deterministic algorithm}\label{subsec:uppcoldet}
We first show that $M'_H$ has rank at most $\prod_{v \in X}(d_{E(H)}(v)+1)$ by exhibiting an explicit factorization. Here we use the shorthand $d_W(v)$ for the number of edges in $W$ containing vertex $v$.
For a bipartite graph~$H$ with parts~$X,Y$ and edges oriented from~$X$ to~$Y$, we have:
\begin{align}
		M'_H[x,y] &= \prod_{(v,w) \in E(H)} (x_v - y_w) \nonumber \\
				  &= \sum_{W \subseteq E(H)} \left (\prod_{v \in X}x^{d_{W}(v)}_v \right ) \left(\prod_{v \in Y} (-y_v)^{d_{E(H)\setminus W}(v)} \right)  \nonumber \\
				  &= \sum_{(d_v \in \{0,\ldots,d_{E(H)}(v)\})_{v\in X}}  \left (\prod_{v \in X}x^{d_v}_v \right )  \left( \sum_{\substack{ W \subseteq E(H) \\ \forall v\in X: d_{W}(v)=d_v }} \prod_{v \in Y} (-y_v)^{d_{E(H)\setminus W}(v)} \right),\label{eq:fac}
\end{align}
where the second equality follows by expanding the product and the third equality follows by grouping the summands on the number of edges incident to vertices in $W$ included in $X$.

Expression~\eqref{eq:fac} provides us with a matrix factorization $M'_H=L_H\cdot R_H$ where $L_H$ is indexed by $x\in [q]^{X}$ and a sequence $s=(d_v \in \{0,\ldots,d_{E(H)}(v)\})_{v\in X}$ and $R_H$ has columns indexed by $y \in [q]^{Y}$ (one such factorization sets $L_H[x,s]=\prod_{v \in X}x^{s_v}_v$). As the number of relevant sequences~$s$ is bounded by~$\prod_{v \in X}(d_{E(H)}(v)+1)$, the factorization implies the claimed rank bound for~$M'_H$.\footnote{This construction (first developed in this paper) has subsequently been used by the second author with Bansal et al.~\cite{DBLP:conf/soda/Bansal0KN18} in the completely different setting of online algorithms; see~\cite[Footnote 3]{DBLP:conf/soda/Bansal0KN18}.} The rank bound allows some partial solutions to be pruned from the dynamic-programming table without changing the answer. The following definition captures correct reduction steps.
\begin{definition}
Fix a bipartite graph $H$ with parts $X$ and $Y$ and let $\cS \subseteq [q]^X$ be a set of $q$-colorings of $X$. We say $\cS' \subseteq [q]^{X}$ \emph{$H$-represents} $\cS$ if $\cS' \subseteq \cS$, and for each $y \in [q]^Y$ we have:\begin{equation}\label{eq:repequiv}
	(\exists x \in \cS\colon x \cup y \text { is a proper coloring of $H$}) \Leftrightarrow (\exists x' \in \cS' \colon x' \cup y \text{ is a proper coloring of $H$}).
\end{equation}
\end{definition}
Note that the backward direction of~\eqref{eq:repequiv} is implied by the property that $\cS' \subseteq \cS$, but we state both for clarity. If $H$ is clear from context it will be omitted. For future reference we record the observation that the transitivity of this relation follows directly from its definition:
\begin{observation}\label{obs:reptrans}
	Let $H$ be a bipartite graph with parts~$X$ and~$Y$, and let~$\mathcal{A},\mathcal{B},\mathcal{C} \subseteq [q]^X$. If $\mathcal{A}$ represents $\mathcal{B}$ and $\mathcal{B}$ represents $\mathcal{C}$, then $\mathcal{A}$ represents $\mathcal{C}$. 
\end{observation}

Given the above matrix factorization, we can directly follow the proof of \cite[Theorem 3.7]{BodlaenderCKN15} to get the following result (note that $\omega$ denotes the matrix multiplication constant):
\begin{lemma}\label{lem:red}
There is an algorithm $\mathtt{reduce}$ that, given a bipartite graph $H$ with parts~$X,Y$ and a set $\cS \subseteq [q]^{X}$, outputs in time $\left(\prod_{v \in X}(d_{E(H)}(v)+1)\right)^{\omega-1}\cdot |\cS| \cdot \poly(|X|+|Y|)$ a set $\cS'$ that represents $\cS$ and satisfies $|\cS'| \leq \prod_{v \in X}(d_{E(H)}(v)+1)$.
\end{lemma}
\begin{proof}
	The algorithm is as follows: compute explicitly the matrix $L_H[\cS,\cdot]$ (i.e.~the submatrix of $L_H$ induced by all rows in $\cS$). As every entry of $L_H$ can be computed in polynomial time, clearly this can be done within the claimed time bound. 
	Subsequently, the algorithm finds a row basis of this matrix and returns that set as $\cS'$. As the rank of a matrix is at most its number of columns, $|\cS'| \leq \prod_{v \in X}(d_{E(H)}(v)+1)$. Using \cite[Lemma 3.15]{BodlaenderCKN15}, this step also runs in the promised running time.
	
	To see that $\cS'$ represents $\cS$, note that clearly $\cS' \subseteq \cS$ and thus it remains to prove the forward implication of~\eqref{eq:repequiv}. To this end, suppose that $x \cup y$ is a proper $q$-coloring of $H$ and $x \in \cS$. As $\cS'$ is a row basis of $L_H$, there exist $x^{(1)},\ldots,x^{(\ell)} \in \cS'$ and $\lambda_1,\ldots,\lambda_\ell$ such that 
	\[
		M'_H[x,y]=L_H[x,\cdot]R_H[\cdot,y]= \left(\sum_{i=1}^\ell \lambda_i L_H[x^{(i)},\cdot]\right) R_H[\cdot,y] = \sum_{i=1}^{\ell} \lambda_i M'_H[x^{(i)},y],
	\]
	where $L_H[x,\cdot]$ and $R_H[\cdot,y]$ denote a row of $L_H$ and column of $R_H$ respectively. As $x \cup y$ is a proper coloring of $H$, Lemma~\ref{lem:zero} implies $M'_H[x,y]$ is non-zero. Therefore there must also exist $x^{(i)} \in \cS'$ such that $M'_H[x^{(i)},y]$ is non-zero and hence $x^{(i)} \cup y$ is a proper coloring of~$H$.
\end{proof}

Equipped with the algorithm $\mathtt{reduce}$ from Lemma~\ref{lem:red} we are ready to present the algorithm for \qColoring. On a high level, the algorithm uses a na\"ive dynamic-programming scheme, but by extensive use of the $\mathtt{reduce}$ procedure we efficiently represent sets of partial solutions and speed up the computation significantly.

First we need to introduce some notation. A vector $x \in V^I$ is \emph{an extension} of a vector $x' \in V^{I'}$ if $I' \subseteq I$ and $x'_i=x_i$ for every $i \in I'$. If $x \in V^I$ and $P \subseteq I$ then the projection $x_{|P}$ is defined as the unique vector in $V^P$ of which $x$ is an extension.
Let $G$ be the graph for which we need to decide whether a proper $q$-coloring exists and fix an ordering $v_1,\ldots,v_n$ of $V(G)$. We denote all edges as directed pairs $(v_i,v_j)$ with $i < j$. For $i=1,\ldots,n$, define $V_i$ as the $i$'th prefix of this ordering, $C_i$ as the $i$'th cut in this ordering, and $X_i$ and $Y_i$ as the left and respectively right endpoints of the edges in this cut, i.e.
\begin{align*}
	V_i &=\{v_1,\ldots,v_i\}, &\qquad C_i &=\{(v_l,v_r) \in E(G): l \leq  i < r\}, \\
	X_i &= \{v_l \in V(G): \exists (v_l,v_r) \in C_i \wedge l < r \}, &\qquad Y_i &= \{v_r \in V(G): \exists (v_l,v_r) \in C_i \wedge l < r \}.
\end{align*}
Note that $X_i \subseteq X_{i-1} \cup \{v_i\}$ and $Y_{i-1} \subseteq Y_i \cup \{v_i\}$. We let $H_i$ denote the bipartite graph with parts~$X_i,Y_i$ and edge set~$C_i$. For $i=1,\ldots,n$, let $T[i] \subseteq [q]^{X_i}$ be the set of all $q$-colorings of the vertices in $X_i$ that can be extended to a proper $q$-coloring of $G[V_i]$. The following lemma shows that we can continuously work with a table $T'$ that represents a table $T$:
\begin{lemma}\label{lem:trans}
If $T'[i-1]$ $H_{i-1}$-represents $T[i-1]$, then $T'[i]$ $H_i$-represents $T[i]$, where
\begin{equation}\label{eq:ti}
	 T'[i] = \left\{ \left(x \cup (v_i,c)\right)_{|X_i} \colon x \in T'[i-1], c \in [q] , \big( \forall v \in N(v_i) \cap X_{i-1} \colon x_{v} \neq c \big) \right\}.
\end{equation}
\end{lemma}
\begin{proof}
Assuming the hypothesis, we first show that $T'[i] \subseteq T[i]$. Let $x \in T'[i-1]$ and $c \in [q]$ such that $\forall v \in N(v_i) \cap X_{i-1} \colon x_{v} \neq c$. As $T'[i-1]$ represents $T[i-1]$, we have that $x \in T[i-1]$. By definition of $T[i-1]$, there exists a proper coloring $w$ of $G[V_{i-1}]$ that extends $x$. Since all $v \in N(v_i) \cap X_{i-1} = N(v_i) \cap V_{i-1}$ satisfy~$x_{v} \neq c$, it follows that $w \cup (v_i,c)$ is a proper coloring of $G[V_i]$, and thus $\left(x \cup (v_i,c)\right)_{|X_i} \in T[i]$.
	
Thus, to prove the lemma it remains to show the forward implication of~\eqref{eq:repequiv}. To this end, let $x \in T[i]$ and let $w \in [q]^{V_i}$ be a proper coloring of $G[V_i]$ that extends $x$. Let $y \in [q]^{Y_i}$ be such that $x \cup y$ is a proper coloring of $H_i$.
As $w_{v_i} \neq w_{v_j}$ for neighbors $v_j \in N(v_i) \cap V_{i-1}$ and $w_{v_i} \neq y_{v_j}$ for $v_j \in N(v_i) \setminus V_i$, it follows that $w \cup y$ extends a proper coloring of $H_{i-1}$.

Therefore $w_{|X_{i-1}} \cup (y \cup (v_i,w_{v_i}))_{|Y_{i-1}}$ must be a proper coloring of $H_{i-1}$, and $w_{|X_{i-1}} \in T[i-1]$ as it can be extended to a proper coloring of $V_i$, and thus also to a proper coloring of $V_{i-1}$. As $T'[i-1]$ $H_{i-1}$-represents $T[i-1]$, there exists $x' \in T'[i-1]$ such that $x' \cup (y \cup (v_i,w_{v_i}))_{|Y_{i-1}}$ is a proper coloring of $H_{i-1}$.

As no neighbor of $v_i$ was assigned color $w_{v_i}$ by $y$, it follows that $(x' \cup (v_i,w_{v_i})) \cup y$ is an extension of a proper coloring of $H_i$. As $x' \cup (y \cup (v,w_{v_i}))_{|Y_{i-1}}$ is a proper coloring of $H_{i-1}$, no neighbors of $v_i$ are assigned color $w_{v_i}$ by $x'$, and by~\eqref{eq:ti} we have that $(x' \cup (v_i,w_{v_i})) \in T'[i]$, as required.
\end{proof}

Now we combine Lemma~\ref{lem:red} with Lemma~\ref{lem:trans} to obtain an algorithm to solve \qColoring.

\begin{lemma}\label{lem:precisetime}
\qColoring can be solved in time $\Oh^*\left(\left(\max_{i}\prod_{v \in X_i}(d_{E(H_i)}(v)+1)\right)^\omega\right)$.
\end{lemma}
\begin{proof}
Note $T'[0]=T[0]=\{\emptyset\}$ (where $\emptyset$ is the $0$-dimensional vector). Using Lemma~\ref{lem:trans}, we can use~\eqref{eq:ti} for $i=1,\ldots,n$ to iteratively compute a set $T'[i]$ representing $T[i]$ from a set $T'[i-1]$ representing $T[i-1]$, and replace $T'[i]$ after each step with $\mathtt{reduce}(H_i,T'[i])$. 
By combining Lemma~\ref{lem:trans} and Observation~\ref{obs:reptrans}, we may conclude that $G$ has a $q$-coloring if and only if $T'[n]$ is not empty (that is, it contains a single element which is the empty vector).

The time required for the computation dictated by~\eqref{eq:ti} is clearly $|T'[i]|\cdot \poly(n)$. Since $|T'[i-1]| \leq \max_{i}\prod_{v \in X_i}(d_{E(H_i)}(v)+1)$, as it is the result of $\mathtt{reduce}$, we have that $|T'[i]|$ is bounded by $q\cdot \max_{i}\prod_{v \in X_i}(d_{E(H_i)}(v)+1)$. Using this upper bound for $T'[i]$, the time of $\mathtt{reduce}$ will be $\Oh^*\left(\left(\max_{i}\prod_{v \in X_i}(d_{E(H_i)}(v)+1)\right)^\omega\right)$, which clearly is the bottleneck in the running time.
\end{proof}

Theorem~\ref{thm:cw} now follows directly from this more general statement.
\begin{proof}[Proof of Theorem~\ref{thm:cw}]
If $v_1,\ldots,v_n$ is a layout of cutwidth $k$, then~$|E(H_i)| \leq k$ for every~$i$, and the term $\prod_{v \in X_i}(d_{E(H_i)}(v)+1)$ is upper bounded by $2^{k}$ by the AM-GM inequality. Thus the theorem follows from Lemma~\ref{lem:precisetime}.
\end{proof}

\subsection{A randomized algorithm}\label{subsec:uppcolrand}
In this section we use an idea similar to the idea from the matrix factorization of the previous section to obtain faster randomized algorithms. Specifically, our main technical result is as follows (recall that $E_i$ denotes the set of edges introduced in bags before $B_i$).
\begin{theorem}\label{thm:qcolpwtech}
	There is a Monte Carlo algorithm for \qColoring that, given a graph $G$ and a nice path decomposition $B_1,\ldots,B_\ell$, runs in time $\Oh^*(\max_{i}\prod_{v \in B_i}(\min\{d_{E_i}(v),d(v)-d_{E_i}(v) \}+1))$. The algorithm does not give false-positives and returns the correct answer with high probability.
\end{theorem}
Let $V(G) = V=\{v_1,\ldots,v_n\}$ be ordered arbitrarily, and direct every edge $\{v_i,v_j\}$ as $(v_i,v_j)$ with $i < j$. Define the \emph{graph polynomial} $f_G$ as $f_G(x_1,\ldots,x_n)=\prod_{(u,v) \in E(G)} (x_u-x_v)$. This polynomial has been studied intensively (cf.~\cite{AlonT97,Loera95,Lovasz82}), for example in the context of the Alon-Tarsi theorem~\cite{AlonT92}. Define $P_G=\sum_{x \in [q]^V}f_G(x)$. Similarly as in Lemma~\ref{lem:zero} we see that if $P_G \neq 0$ then $G$ has a proper $q$-coloring, and if $G$ has a unique $q$-coloring then $P_G \neq 0$ as it is the product of non-zero values. This is useful if the graph is guaranteed to have at most one proper $q$-coloring. To this end, we use a standard technique based on the Isolation Lemma, which we state now.

\begin{definition}
	A function $\omega\colon U \rightarrow \mathbb{Z}$ \emph{isolates} a set family $\mathcal{F} \subseteq 2^U$ if there is a unique $S' \in \mathcal{F}$ with $\omega(S')= \min_{S \in \mathcal{F}}\omega(S)$, where~$\omega(S') := \sum _{v \in S'} \omega(v)$.
\end{definition}

\begin{lemma}[Isolation Lemma, \cite{MulmuleyVV87}]
	\label{lem:iso}
	Let $\mathcal{F} \subseteq 2^U$ be a non-empty set family over universe~$U$.
	For each $u \in U$, choose a weight $\omega(u) \in \{1,2,\ldots,W\}$ 
	uniformly and independently at random.
	Then 
	$\Pr[\omega \textnormal{ isolates } \mathcal{F}] \geq 1 - |U|/W$.
\end{lemma}

We will apply Lemma~\ref{lem:iso} to isolate the set of proper colorings of $G$. To this end, we use the set~$V(G) \times [q]$ of vertex/color pairs as our universe~$U$, and consider a weight function~$\omega \colon V(G) \times [q] \to \mathbb{Z}$.

\begin{definition}
A \emph{$q$-coloring} of $G$ is a vector $x \in [q]^n$, and it is proper if $x_i\neq x_j$ for every $(i,j) \in E(G)$. The \emph{weight} of $x$ is $\omega(x)=\sum_{i = 1}^n\omega((i,x_i))$.
\end{definition}

Let $\omega \colon V(G) \times [q] \rightarrow [2nq]$ be a random weight function, i.e.~for every $v \in V(G)$ and $c \in [q]$ we pick an integer from $[2nq]$ uniformly and independently at random. For every integer $z$ we associate a number $P_G(z)$ with $G$, as follows:
\begin{equation}\label{eq:poly}
P_G(z)=\sum_{\substack{x \in [q]^{n} \\ \omega(x)=z}}\prod_{(i,j) \in E(G)}\left( x_i-x_j\right). 
\end{equation}
If $G$ has no proper $q$-coloring, then $P_G(z)=0$ since for every $q$-coloring $x$ there will be an edge $(i,j) \in E$ for which $x_i=x_j$ and therefore the product in~\eqref{eq:poly} vanishes. We claim that if $G$ has a proper $q$-coloring, then with probability at least~$1/2$ there exists $z\leq 2qn$ such that $P_G(z) \neq 0$, which means we get a correct algorithm with high probability by repeating a polynomial in~$n$ number of times. Let $\mathcal{F} = \{ \{(i,x_i)\}_{i \in V} : x \text{ is a proper $q$-coloring of $G$} \} \subseteq 2^U$. As $\mathcal{F}$ is non-empty, we may apply Lemma~\ref{lem:iso} to obtain that $\omega$ isolates $\mathcal{F}$ with probability at least $1/2$. Conditioned on this event, there must exist an integer $w$ such that there is exactly one proper $q$-coloring $x$ of $G$ satisfying $\omega(x)=z$. In this case, $x$ is the only summand in~\eqref{eq:poly} that can have a non-zero contribution. Moreover, as it is a proper coloring, its contribution is a product of non-zero entries and therefore non-zero itself. Thus $P_G(z)$ is non-zero with probability at least~$1/2$.

We now continue by showing how to compute $P_G(z)$ for all $z \leq 2qn$ quickly using dynamic programming. Note that by expanding the product in~\eqref{eq:poly} we have:
\begin{equation}\label{eq:poly2}
	P_G(z) = \sum_{\substack{x \in [q]^{n} \\ \omega(x)=z}}\sum_{W \subseteq E(G)}\left(\prod_{(u,v) \in W}x_u \right) \left(\prod_{(u,v) \in E(G) \setminus W}-x_v \right).
\end{equation}
If $B_i$ is a bag of a path decomposition (Section~\ref{sec:preliminaries}), we need to define table entries $T_i$ containing all information about the graph $(V_i,E_i)$ needed to compute $P_G(z)$. 
Before we describe these table entries we make a small deviation to convey intuition about our approach. Specifically, we may interpret $P_G(z)$ as a polynomial in variables $x_v$ for $v \in B_i$. Now suppose for simplicity that $|B_i|=1$. Then the amount of information about $E_i$ needed to compute $P_G(z)$ may be studied via a simple communication-complexity game that we now outline.

\subparagraph*{A One-way Communication Protocol}
Alice has a univariate polynomial $P_A(x)$ of degree $d_A$, and Bob has a univariate polynomial $P_B(x)$ of degree $d_B$. Both parties know $d_A,d_B$ and an additional integer $q$. Alice needs to send as few bits as possible to Bob after which Bob needs to output the quantity $\sum_{x \in [q]} P_A(x) P_B(x)$, where $q \in \mathbb{N}$ is known to both.

An easy strategy is that Alice sends the $d_A+1$ coefficients of her polynomial to Bob. An alternative strategy for Alice is based on partial evaluations, which is useful when~$d_B < d_A$. By expanding Bob's polynomial in coefficient form we can rewrite $\sum_{x \in [q]} P_A(x) P_B(x)$ into
\[
 \sum_{x \in [q]} P_A(x) (c_0x^0+c_1x^1+\ldots+c_{d_B}x^{d_B}) = c_0 \sum_{x \in [q]} P_A(x) x^0 + \ldots + c_{d_B} \sum_{x \in [q]} P_A(x) x^{d_B},
\]
so as second strategy Alice may send the $d_B+1$ values $\sum_{x \in [q]} P_A(x) x^i$ for $i=0,\ldots,d^B$. So she can always send at most $\min\{d_A,d_B\}+1$ integers.

In our setting for defining table entries $T_i$ for evaluating $P_G(z)$, we think of $d_A(v)$ as the number of edges in $E_i$ incident to $v$ and of $d_B(v)$ as the number of edges incident to $v$ not in $E_i$. Roughly speaking, the running time of Theorem~\ref{thm:qcolpwtech} is obtained by defining table entries storing Alice's message, in which she chooses the best of the two strategies independently for every vertex.

\subparagraph*{Definition of the Table Entries}
An \emph{orientation} $O$ of a subset $X \subseteq E(G)$ of edges is a set of directed pairs such that for every $\{u,v\} \in X$, either $(u,v) \in O$ or $(v,u) \in O$. If $O$ is an orientation of $X$, we also say $O$ \emph{orients} $X$. The number of \emph{reversals} $\mathrm{rev}(O)$ of $O$ is the number of $(v,u) \in O$ such that $u$ is introduced in a bag before the bag in which $v$ is introduced. An orientation is \emph{even} if its number of reversals is even, and it is \emph{odd} otherwise. 

For a fixed path decomposition~$B_1, \ldots, B_\ell$ of the input graph~$G$, let $L_i \subseteq B_i$ consist of all vertices in $B_i$ of which at most half of their incident edges are already introduced in~$B_i$ or a bag before $B_i$, and let $R_i = B_i \setminus L_i$.
Let $l^i$ be the vector indexed by $L_i$ such that for every $v \in L_i$ the value $l^i_v$ denotes the number of edges incident to $v$ already introduced before or at bag $B_i$. Similarly, let $r^i$ be the vector indexed by $R_i$ such that for every $v \in R_i$ the value $r^i_v$ denotes the number of edges incident to $v$ introduced \emph{after} bag~$B_i$. So for every~$i$ we have $d(v)=l^i_v+r^i_v$.

If $b \in \mathbb{N}_{\geq 0}^{I}$ is a vector, we denote $\mathcal{P}(b)$ for the set of vectors $a$ in $\mathbb{N}_{\geq 0}^{I}$ such that $a\preceq b$. Here $a \preceq b$ denotes that $a_v \leq b_v$ for every $v \in I$.  
For $d \in \mathcal{P}(l^i)$ and $e \in \mathcal{P}(r^i)$, define:
\begin{equation}\label{eq:tableentries}
T^z_i[d,e] = \sum_{\substack{x \in [q]^{V_i \setminus L_i} \\ \omega(x)=z}} \sum_{\substack{O \text{ orients } E_i \\ \forall u\in L_i: d^+_O(u)=d_u}} (-1)^{\mathrm{rev}(O)} \left( \prod_{u \in V_i \setminus L_i} x^{d^+_O(u)}_u \right)  \left(\prod_{u \in R_i} x_u^{e_u} \right).
\end{equation}

Intuitively, this could be seen as a partial evaluation of $P_G(z)$. Note we sum over all possible $x_v \in [q]$ for $v \in V_i \setminus L_i$, but let the values $x_v$ for $v \in L_i$ be undetermined and store the coefficient in the obtained polynomial of a certain monomial $\prod_{u \in R_i} x_u^{e_u}$. Indeed, it is easily seen that $P_G(z)$ equals $T^z_\ell[\emptyset,\emptyset]$, where $\emptyset$ is the unique $0$-dimensional vector. By combining the appropriate recurrence for all values $T^z_i[d,e]$ with dynamic programming, the following lemma is proved in Appendix~\ref{sec:rec}.

\begin{lemma}\label{lem:rec}All values $T^z_{i}[d,e]$ can be computed in time $\poly(n) \cdot \sum_{i=1}^\ell T_i$, where
\[
T_i = |\mathcal{P}(l^i)| \cdot |\mathcal{P}(r^i)|=\prod_{v \in B_i}(\min\{d_{E_i}(v),d(v)-d_{E_i}(v) \}+1).
\] 
\end{lemma}

Thus $P_G(z)$ can be computed in the time stated in Theorem~\ref{thm:qcolpwtech}. As discussed, $P_G(z)=0$ if $G$ has no proper $q$-coloring. Otherwise, $\omega$ isolates the set of proper $q$-colorings of $G$ with probability at least $1/2$. Conditioned on this event we have $P_G(z) \neq 0$, where $z$ is the weight of the unique minimum-weight $q$-coloring. Therefore we output \textsc{yes} if $P_G(z) \neq 0$ for some~$z$ and obtain the claimed probabilistic guarantee. This concludes the proof of Theorem~\ref{thm:qcolpwtech}. 

As special cases of Theorem~\ref{thm:qcolpwtech} we obtain Theorems~\ref{thm:coloring:randomized} and~\ref{thm:alg:degree}.

\begin{proof}[Proof of Theorem~\ref{thm:coloring:randomized}.]
Given a linear layout $v_1,\ldots,v_n$ of cutwidth~$k$, define a nice path decomposition in which vertices are introduced in the order of the layout. After~$v_i$ is introduced, its incident edges to~$v_j$ with~$j < i$ are introduced in arbitrary order. Forget~$v_i$ directly after the series of edge introductions that introduced its last incident edge.
	
As $v_1,\ldots,v_n$ has cutwidth at most $k$, for any bag $B_i$ of this path decomposition the number of edges between $V_i$ and $V \setminus V_i$ is at most $k$.  Together with the edges incident on the most-recently introduced vertex~$v_j$, these $k$ edges are the only edges incident on~$B_i$ that are not in $E_i$. Consider the term~$\prod_{v \in B_i}(\min\{d_{E_i}(v),d(v)-d_{E_i}(v) \}+1)$. Vertex~$v_j$ contributes at most one factor~$n$. For the remaining vertices in~$B_i$, the only incident edges not in~$E_i$ are those in the cut of size at most~$k$. By the AM-GM inequality, their contribution to the product is maximized when they are all incident to distinct vertices, in which case the algorithm of Theorem~\ref{thm:qcolpwtech} runs in time $\Oh^*(2^k)$.
\end{proof}	

\begin{proof}[Proof of Theorem~\ref{thm:alg:degree}]
Follows from Theorem~\ref{thm:qcolpwtech}: $\min\{d_{E_i}(v),d(v)-d_{E_i}(v)\} \leq \lfloor d(v)/2 \rfloor$.
\end{proof}

\section{Lower Bounds for Graph Coloring}\label{sect:lowmain}
In this section we discuss the main ideas behind our lower bounds, whose proofs are deferred to the appendix. We first start with Theorem~\ref{thm:3collow}, which rules out algorithms for solving \textsc{$3$-Coloring} in time~$\Oh^*((2-\varepsilon)^\cutw)$, even on \emph{planar graphs}. (We remark that a companion paper~\cite{GeffenJKM18} was the first to present lower bounds for planar graphs of bounded cutwidth.) The overall approach is based on the framework by Lokshtanov et al.~\cite{LokshtanovMS11}. We prove that an $n$-variable instance of CNF-SAT can be transformed in polynomial time into an equivalent instance of \textsc{$3$-Coloring} on a planar graph~$G$ with a linear layout of cutwidth~$n + \Oh(1)$. Consequently, saving~$\varepsilon$ in the base of the exponent when solving graph coloring would violate SETH. By employing clause-checking gadgets in the form of a path~\cite{JaffkeJ17}, crossover gadgets~\cite{GareyJS76}, and a carefully constructed ordering of the graph, we get the desired reduction.

The second lower bound, Theorem~\ref{thm:degreecollow}, rules out algorithms with running time~$\Oh^*((\lfloor d/2 \rfloor + 1 - \varepsilon)^{\pw})$ for solving~$q$-\textsc{Coloring} for $q := \lfloor d /2 \rfloor + 1$ on graphs of maximum degree~$d$ and pathwidth~$\pw$, for any odd integer~$d \geq 5$. The reduction employs \emph{chains of cliques} to propagate assignments throughout a bounded-pathwidth graph. A $t$-chain of $q$-cliques is the graph obtained from a sequence of~$t$ vertex-disjoint $q$-cliques by selecting a distinguished \emph{terminal} vertex in each clique and connecting it to the $(q-1)$ non-terminals in the previous clique. Any proper $q$-coloring of a chain assigns all terminals the same color, and terminals have~$2(q-1)$ neighbors in the chain. Therefore, we can propagate a choice with~$q$ possibilities throughout a path decomposition. We encode truth assignments to variables of a \textsc{CNF-SAT} instance through colors given to the terminals of such chains. We enforce that the encoded truth assignment satisfies a clause, by enforcing that an assignment that does \emph{not} satisfy the clause, is not the one encoded by the coloring. To check this, we take one terminal from each chain and connect it to a partner on a path gadget that forbids a specific coloring. Hence each vertex on a chain will receive at most one more neighbor, giving a maximum degree of~$d := 2(q-1)+1 = 2q-1$ to represent a \qColoring instance. Then solving this \qColoring instance in~$\Oh^*((\lfloor d/2 \rfloor + 1 - \varepsilon)^\pw) = \Oh^*(((q-1) + 1 - \varepsilon)^\pw)$ time will contradict SETH for the same reason as in the earlier construction~\cite{LokshtanovMS11} showing the impossibility of~$\Oh^*((q-\varepsilon)^\pw)$-time algorithms.

\section{Conclusion}\label{sec:conc}

We showed how graph decompositions using small edge separators can be used to solve \textsc{$q$-Coloring}. The exponential parts of the running times of our algorithms are independent of~$q$, which is a significant difference compared to algorithms for parameterizations based on vertex separators. The deterministic~$\Oh^*(2^{\omega\cdot \cutw})$ algorithm of Theorem~\ref{thm:cw} for the cutwidth parameterization follows cleanly from the bound on the rank of the partial solutions matrix. It may serve as an insightful new illustration of the rank-based approach for dynamic-programming algorithms in the spirit of~\cite{BodlaenderCKN15,CyganKN13,CyganNPPRW11,FominLPS16}.

One of the main take-away messages from this work from a practical viewpoint is the following. Suppose~$H$ is a subgraph of~$G$ connected to the remainder of the graph by~$k$ edges. Then any set of partial colorings~$\mathcal{S}$ of~$H$ can be reduced to a subset~$\mathcal{S'}$ of size~$2^k$, with the guarantee that if some coloring in~$\mathcal{S}$ could be extended to a proper coloring of~$G$, then this still holds for~$\mathcal{S'}$. The reduction can be achieved by an application of Gaussian elimination, which has experimentally been shown to work well for speeding up dynamic programming for other problems~\cite{FafianieBN15}. We therefore believe the table-reduction steps presented here may also be useful when solving graph coloring over tree- or path decompositions, and can be applied whenever processing a separator consisting of few edges.

\bibliography{Cutwidth}

\clearpage

\appendix

\section{Proof of Lemma~\ref{lem:rec}: A Recurrence for Computing Table Entries for the Randomized Algorithm}\label{sec:rec}

In this section we provide a recurrence to compute table entries as defined in~\eqref{eq:tableentries}. We make a case distinction based on the bag type:

\begin{description}
	\item[\textbf{Left-most Bag:}] If $B_i$ is a leaf bag (i.e.~$i=1$ and~$V_i = \emptyset$), then~$V_i \setminus L_i$ has one $q$-coloring of weight $0$ and orientation of the empty set which amounts to one summand which contributes $1$, so
	\[
	T^z_1[\emptyset,\emptyset] =
	\begin{cases}
	1, & \text{if } z=0\\
	0, & \text{otherwise}. 
	\end{cases}
	\]
	where $\emptyset$ denotes the unique $0$-dimensional vector.
	\item[\textbf{Introduce Vertex Bag:}] If $v$ is the vertex introduced in bag $B_{i}$, then~$v \in L_{i}$ as $l^{i}_v=0$ and we see that $T^z_{i}[d \cup \{(v,0)\},e]$ equals
	\begin{align*}
	& \sum_{\substack{x \in [q]^{V_{i-1} \setminus L_{i-1}}\\ \omega(x)=z}} \sum_{\substack{O \text{ orients } E_{i-1} \\ \forall u\in L_{i-1}: d^+_O(u)=d_u \\ d^+_O(v)=0}} (-1)^{\mathrm{rev}(O)} \left( \prod_{u \in V_{i-1} \setminus L_{i-1}} x^{d^+_O(u)}_u \right)  \left(\prod_{u \in R_{i-1}} x_u^{e_u} \right)\\
	&= T^z_{i-1}[d,e].
	\end{align*}
	\item[\textbf{Forget Vertex Bag:}] If $v$ is the vertex forgotten in bag $B_{i}$, then~$v \in R_{i-1}$ as $G$ has no isolated vertices, implying~$L_i = L_{i-1}$ and~$R_i = R_{i-1} \setminus \{v\}$. As~$B_i$ is a forget bag, we have~$V_i = V_{i-1}$ and~$E_i = E_{i-1}$. Since~$x^0_v = 1$ for any~$x_v\in [q]$, we see that~$T^z_{i}[d,e]$ equals
	\begin{align*}
	&\sum_{\substack{x \in [q]^{V_{i} \setminus L_{i}}\\ \omega(x)=z}} \sum_{\substack{O \text{ orients } E_{i} \\ \forall u\in L_{i}: d^+_O(u)=d_u}} (-1)^{\mathrm{rev}(O)} \left( \prod_{u \in V_{i} \setminus L_{i}} x^{d^+_O(u)}_u \right) \left(\prod_{u \in R_{i}} x_u^{e_u} \right) x^0_v\\ 
	&=\sum_{\substack{x \in [q]^{V_{i-1} \setminus L_{i-1}}\\ \omega(x)=z}} \hspace{3em} \sum_{\mathclap{\substack{O \text{ orients } E_{i-1} \\ \forall u\in L_{i-1}: d^+_O(u)=d_u}}} (-1)^{\mathrm{rev}(O)} \left( \prod_{u \in V_{i-1} \setminus L_{i-1}} x^{d^+_O(u)}_u \right) \left(\prod_{u \in R_{i-1}} x_u^{(e \cup \{(v,0)\})_u} \right)\\
	&= T^z_{i-1}[d,e \cup \{(v,0)\}].
	\end{align*}
	
	\item[\textbf{Introduce Edge Bag:}] Suppose the edge $(v,w)$ is introduced in bag $B_{i}$, so that~$V_i = V_{i-1}$. Let $f_1$ denote the edge $(v,w)$ and $f_2$ denote its reversal $(w,v)$. For $j \in \{1,2\}$ define 	
	\[
	T^z_{i,j}[d,e] = \sum_{\substack{x \in [q]^{V_i \setminus L_{i}}\\ \omega(x)=z}} \sum_{\substack{O \text{ orients } E_{i} \\ f_j \in O \\ \forall u\in L_{i}: d^+_O(u)=d_u}} (-1)^{\mathrm{rev}(O)} \left( \prod_{u \in V_i \setminus L_{i}} x^{d^+_O(u)}_u \right)  \left(\prod_{u \in R_{i}} x_u^{e_u} \right).
	\]
	We have $T^z_i[d,e]=T^z_{i,1}[d,e]+T^z_{i,2}[d,e]$, so it remains to compute the latter two expressions. We will focus on how to compute $T^z_{i,1}[d,e]$ as computing $T^z_{i,2}[d,e]$ is done similarly by replacing $v$ with $w$, replacing~$f_1$ by~$f_2$, and multiplying by $-1$ (to account for $\mathrm{rev}(O)$). We distinguish the following cases:
	
	\begin{description}
		\item[\textbf{$v$ is in $L_i$:}] As $f_1$ contributes one to $d^+_O(v)$, we have
		\begin{align*}
		T^z_{i,1}[d \cup \{(v,d_v)\},e] &= T^z_{i-1}[d \cup \{(v,d_v-1)\},e].
		\end{align*}
		\item[\textbf{$v$ is in $R_i \cap R_{i-1}$:}] As $f_1$ contributes one to $d^+_O(v)$ we see that $T^z_{i,1}[d,e] = T^z_{i-1}[d,e]x_v$ (interpreting~$x_v$ as being bound by the quantification in the expression for~$T^z_{i-1}[d,e]$) so that we obtain
		\begin{align*}
		T^z_{i,1}[d,e \cup \{(v,d_v)\}] &= T^z_{i-1}[d,e \cup \{(v,d_v+1)\}].
		\end{align*}
		\item[\textbf{$v$ is in $R_i \cap L_{i-1}$:}] This is the most complicated case, which occurs when the introduction of the edge~$(v,w)$ changes the status of vertex~$v$ from being a vertex in~$L_{i-1}$ of which at most half the incident edges are introduced, to a vertex in~$R_i$ of which more than half its incident edges are introduced. The expression for~$T^z_{i,1}[d,e]$ sums over the colorings of vertices of the set~$V_i \setminus L_i$ which contains~$v$, while the expressions for~$T_{i-1}$ sum over colorings of the set~$V_{i-1} \setminus L_{i-1}$ which does not contain~$v$. Moreover, the expression for~$T^z_{i,1}[d,e]$ sums over all orientations of~$E_i$, regardless of the out-degree of~$v$ in that orientation, and includes a factor~$x_v^{d^+_O(v)}$ in the first product, while the expressions for~$T_{i-1}$ only sum over orientations in which the out-degree of~$v \in L_{i-1}$ is as specified by~$d_v$, and contains no such factor. Finally, the expression for~$T^z_{i,1}[d,e]$ contains a term~$x_v^{e_v}$ in the second product since~$v \in R_i$, but the expressions for~$T_{i-1}$ do not. So to obtain~$T^z_{i,1}[d,e]$ from the values of~$T_{i-1}$, we sum over all relevant choices of out-degree~$d_v$ and color~$x_v$, and multiply the values of the appropriate subproblems by~$x_v^{e_v + d_v + 1}$. The final~$+1$ stems from the fact that~$f_1 \in O$ contributes one to~$d^+_O(v)$. 
		
		With this in mind, observe that~$T^z_{i,1}[d,e]$ equals
		\begin{align*}
\hspace{-4em}&\sum_{\substack{0 \leq d_v \leq l^{i-1}_v \\ x_v \in [q]}}\sum_{\substack{ x \in [q]^{V_{i-1} \setminus L_{i-1}} \\ \omega(x)+\omega((v,x_v))=z}} \hspace{3em} \sum_{\mathclap{\substack{O \text{ orients } E_{i} \\ f_1 \in O \\ \forall u\in L_{i-1}: d^+_O(u)=d_u}}} (-1)^{\mathrm{rev}(O)}\hspace{-0.3em}\left( \prod_{v \in V_{i-1} \setminus L_{i-1}} x^{d^+_O(u)}_u \right)\hspace{-0.5em}\left(\prod_{u \in R_{i-1}} x_u^{e_u} \right)\hspace{-0.3em}x^{e_v+d_v+1}_v\\
		\hspace{-5em}& \hspace{10em}  = \sum_{0 \leq d_v \leq l^{i-1}_v}\sum_{x_v \in [q]} T^{z-\omega((v,x_v))}_{i-1}[d \cup \{(v,d_v)\},e \setminus \{(v,e_v)\}]x^{e_v+d_v+1}_v.		
		\end{align*}
	\end{description}
\end{description}

Recall the range of the weight function~$\omega \colon V(G) \times [q] \rightarrow [2nq]$. Using the above recurrence, we can compute $P_G(z)$ for every $0 \leq z \leq 2nq$ in time $\poly(n) \cdot \sum_{i=1}^\ell T_i$ where 
\[
T_i = |\mathcal{P}(l^i)| \cdot |\mathcal{P}(r^i)|=\prod_{v \in B_i}(\min\{d_{E_i}(v),d(v)-d_{E_i}(v) \}+1).
\]

This concludes the proof of Lemma~\ref{lem:rec}.

\section{Lower bounds for graph coloring}\label{sect:low}
\newcommand{\qdColoring}{\textsc{$q_d$-Coloring}\xspace}
\newcommand{\sSAT}{\textsc{$s$-SAT}\xspace}
\newcommand{\Hvc}{\ensuremath{H_{\mathrm{\textsc{vc}}}}\xspace}
\newcommand{\Hds}{\ensuremath{H_{\mathrm{\textsc{vc}}}}\xspace}
\newcommand{\Hcol}{\ensuremath{H_{\mathrm{\textsc{col}}}}\xspace}

\subsection{Lower bounds for 3-coloring on planar graphs of bounded cutwidth}
In this section we prove that \textsc{$3$-Coloring} on planar graphs of cutwidth~$k$ cannot be solved in~$\Oh^*((2-\varepsilon)^k)$ time for any~$\varepsilon > 0$, unless SETH fails. To lift our hardness result to planar graphs, we will employ ideas from the NP-completeness proof of \textsc{$3$-Coloring on planar graphs} by Garey, Johnson, and Stockmeyer~\cite[Thm. 2.2]{GareyJS76}. It relies on the gadget graph~$\Hcol$ due to Michael Fischer that is shown in Figure~\ref{fig:threecoloring:gadget}.

\begin{observation}
Any proper $3$-coloring~$f \colon V(\Hcol) \to [3]$ of~$\Hcol$ satisfies~$f(u) = f(u')$ and~$f(v) = f(v')$. Conversely, any coloring~$f' \colon \{u,u',v,v'\} \to [3]$ with~$f'(u) = f'(u')$ and~$f'(v) = f'(v')$ can be extended to a proper $3$-coloring of~$\Hcol$.
\end{observation}

We introduce some terminology for working with planar graphs, as well as drawings of non-planar graphs. A \emph{drawing} of a graph~$G$ is a function~$\psi$ that assigns a unique point~$\psi(v) \in \mathbb{R}^2$ to each vertex~$v \in V(G)$, and a curve~$\psi(e) \subseteq \mathbb{R}^2$ to each edge~$e \in E(G)$, such that the following four conditions hold. (1) For~$e = \{u,v\} \in E(G)$, the endpoints of~$\psi(e)$ are exactly~$\psi(u)$ and~$\psi(v)$. (2) The interior of a curve~$\psi(e)$ does not contain the image of any vertex. (3) No three curves representing edges intersect in a common point, except possibly at their endpoints. (4) The interiors of the curves~$\psi(e), \psi(e')$ for distinct edges intersect in at most one point. If the interiors of all the curves representing edges are pairwise-disjoint, then we have a \emph{planar drawing}. We will combine (non-planar) drawings with crossover gadgets to build planar drawings. A graph is \emph{planar} if it admits a planar drawing. 

The operation of \emph{identifying vertices~$u$ and~$v$} in a graph~$G$ results in the graph~$G'$ that is obtained from~$G$ by replacing the two vertices~$u$ and~$v$ by a new vertex~$w$ with~$N_{G'}(w) = N_G(\{u,v\})$. 

\begin{figure}[t]
\centering
\includegraphics{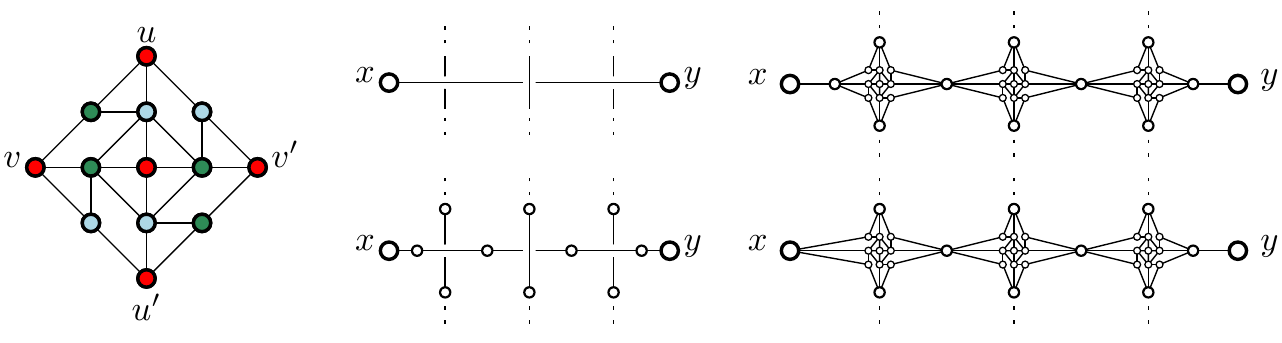}
\caption{Illustration for planarizing \textsc{$3$-Coloring}. Left: the graph~$\Hcol$ with a $3$-coloring in which all terminals are red. The rest of the figure illustrates Theorem~\ref{thm:planarize:threecol} on the crossed edge~$\{x,y\}$ in the top middle. The results of steps~\ref{step:threecol:subdivide},~\ref{step:threecol:insert}, and~\ref{step:threecol:merge} are shown in the middle bottom, top right, and bottom right figures, respectively.}
\label{fig:threecoloring:gadget}
\end{figure}

The following theorem is illustrated in Figure~\ref{fig:threecoloring:gadget}.

\begin{theorem}[{\cite[Thm 2.2]{GareyJS76}}] \label{thm:planarize:threecol}
Let~$G$ be a graph with a drawing~$\psi$, and define~$G'$ as follows:
\begin{enumerate}
	\item For each edge~$e = \{x,y\} \in E(G)$, consider the curve~$\psi(e)$. If the interior of~$\psi(e)$ is crossed, then insert new vertices on the curve~$\psi(e)$ between the endpoints of~$\psi(e)$ and their nearest crossing, and between consecutive crossings along~$\psi(e)$.\label{step:threecol:subdivide}
	\item Replace each crossing in the drawing by a copy of the graph~$\Hcol$, identifying the terminals~$u$ and~$u'$ with the nearest newly inserted vertices on either side along one curve, and identifying~$v$ and~$v'$ with the nearest newly inserted vertices on the other curve.\label{step:threecol:insert}
	\item For each~$\{x,y\} \in E(G)$, choose one endpoint as the \emph{distinguished endpoint} and identify it with the nearest new vertex on the curve~$\psi(\{u,v\})$.\label{step:threecol:merge}
\end{enumerate}
Then the resulting graph~$G'$ is planar, and it has a proper $3$-coloring if and only if~$G$ has one.
\end{theorem}

Using Theorem~\ref{thm:planarize:threecol} we present our lower bound for solving \textsc{$3$-Coloring} on planar graphs of bounded cutwidth. The following lemma is the main ingredient.

\begin{lemma} \label{lemma:coloring:lb:construction}
There is a polynomial-time algorithm that, given a CNF formula~$\phi$ on~$n$ variables, constructs a planar graph~$G$ together with a linear layout~$\pi$ with~$\cutw_\pi(G) \leq n + \Oh(1)$, such that~$\phi$ is satisfiable if and only if~$G$ is 3-colorable.
\end{lemma}
\begin{proof}
The construction is split into three steps. (I) In the first step we transform~$\phi$ into a non-planar \textsc{List-$3$-Coloring} instance~$(G_1, L_1)$ of cutwidth~$n + \Oh(1)$. In this list-coloring instance, each vertex~$v \in V(G_1)$ is assigned a list~$L_1(v) \subseteq [3]$ of allowed colors. The question is whether~$G_1$ has a proper $3$-coloring~$f \colon V(G_1) \to [3]$ such that~$f(v) \in L_1(v)$ for each~$v \in V(G_1)$. We will ensure that such a list coloring exists if and only if~$\phi$ is satisfiable. (II) In the second step we transform the \textsc{List-$3$-Coloring} instance~$(G_1, L_1)$ into an equivalent instance~$G_2$ of the plain \textsc{$3$-Coloring} problem, while maintaining a bound of~$n + \Oh(1)$ on the cutwidth of~$G_2$. (III) Finally, we draw~$G_2$ in a particular way and use Theorem~\ref{thm:planarize:threecol} to turn it into a planar graph~$G_3$ without blowing up the cutwidth.

\subparagraph{(I) Constructing a list-coloring instance}
This part is almost identical to a construction in earlier work~\cite[Thm. 2]{JansenK13} giving kernelization lower bounds for structural parameterizations of graph coloring problems. We repeat it here because the remainder of the proof builds on it. Let~$\phi$ be a CNF formula on variables~$x_1, \ldots, x_n$ with clauses~$C_1, \ldots, C_m$. We assume without loss of generality that no clause of~$\phi$ contains repeated literals, since repetitions can be omitted without changing the satisfiability status. We also assume that no clause contains a literal and its negation (such clauses are trivially satisfied). For~$j \in [m]$ we use~$|C_j|$ to denote the number of literals in~$C_j$, each of which is of the form~$x_i$ or~$\neg x_i$ for some~$i \in [n]$. Construct a graph~$G_1$ and a list function~$L_1 \colon V(G_1) \to 2^{[3]}$ as follows (see Figure~\ref{fig:listthreecol:cw:lb}).

\begin{figure}[t]
\centering
\includegraphics{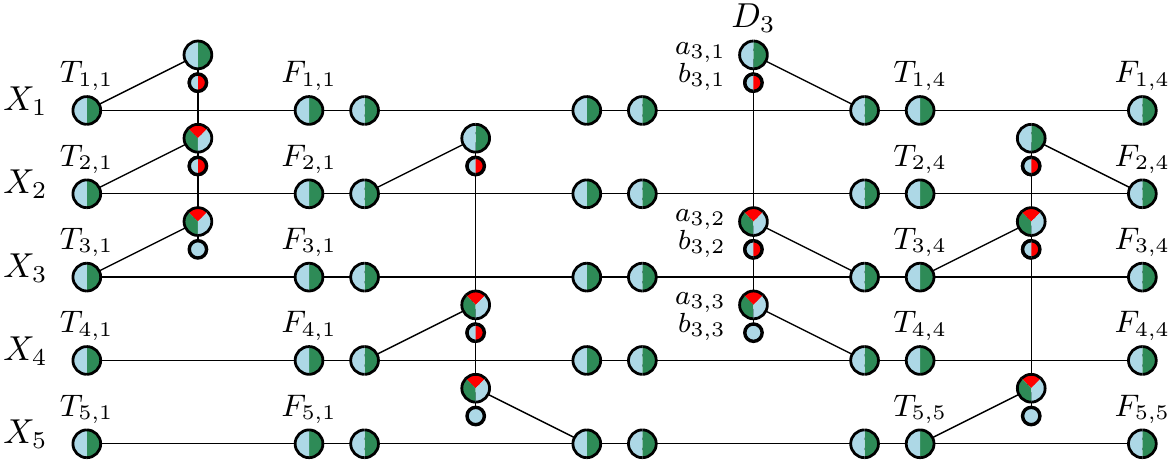}
\caption{Illustration of the construction of Lemma~\ref{lemma:coloring:lb:construction} applied to the formula~$\phi = (x_1 \vee x_2 \vee x_3) \wedge (x_2 \vee x_4 \vee \neg x_5) \wedge (\neg x_1 \vee \neg x_3 \vee \neg x_4) \wedge (\neg x_2 \vee x_3 \vee x_5)$. The \textsc{List-$3$-Coloring} instance~$(G_1,L_1)$ that is constructed for~$\phi$ is drawn. The colors available to a vertex are drawn inside the circle that represents it, using the encoding that~$1$ is red,~$2$ is blue, and~$3$ is green. Observe that a clause path such as~$D_3$ cannot be properly colored without using color~$3$ (green).}
\label{fig:listthreecol:cw:lb}
\end{figure}

\begin{enumerate}
	\item For each~$i \in [n]$, add to~$G_1$ a path~$X_i$ consisting of~$2m$ consecutive vertices~$T_{i,1}, F_{i,1}, T_{i,2}, \linebreak[1] \ldots, \linebreak[1] T_{i,m}, F_{i,m}$ representing \true and \false, respectively. Give each vertex on this path the list consisting of~$\{2,3\}$. We refer to~$X_i$ as the \emph{variable path for~$x_i$}. In a proper list coloring, all vertices~$\{T_{i, j} \mid j \in [m]\}$ have the same color, and all vertices~$\{F_{i,j} \mid j \in [m]\}$ have the same color. Coloring the $T$-vertices~$2$ and the $F$-vertices~$3$ encodes that variable~$x_i$ is \true, while coloring the $T$-vertices~$3$ and the $F$-vertices~$2$ encodes that~$x_i$ is \false. Observe that by this interpretation, the vertices representing literals that evaluate to \true are colored~$2$ and the literals evaluating to \false are colored~$3$.
	\item For each~$j \in [m]$, add to~$G_1$ a path~$D_j$ consisting of~$2|C_j|$ consecutive vertices~$a_{j,1}, b_{j,1}, \ldots, \linebreak[1] a_{j, |C_j|}, b_{j, |C_j|}$. We refer to~$D_j$ as the \emph{clause path} for~$C_j$. Set the lists of its vertices as follows. All $a$-vertices have the full list~$\{1,2,3\}$, except the first one for which~$L_1(a_{j,1}) := \{2,3\}$. All $b$-vertices have the list~$\{1,2\}$, except the last one for which~$L_1(b_{j,|C_j|}) := \{2\}$. Observe that in any proper list coloring of the path~$D_j$, one of the vertices has to receive color~$3$. If no vertex has color~$3$, then the lists enforce that~$a_{j,1}$ and~$b_{j,|C_j|}$ both have color~$2$ while the colors must alternate between~$1$ and~$2$ along the path; this implies that two adjacent vertices receive identical colors since the path contains an even number of vertices.
	\item As the last part of the construction we connect the clause paths to the variable paths to enforce that clauses are satisfied. For each~$j \in [m]$, let~$\ell_{i_1}, \ell_{i_2}, \ldots, \ell_{i_{|C_j|}}$ be the literals occurring in clause~$C_j$, and sort them such that~$i_1 < i_2 < \ldots < i_{|C_j|}$. For each~$k \in [|C_j|]$, do the following. If~$\ell_{i_k}$ is a positive occurrence of variable~$x_{i_k}$, then make~$a_{j,k}$ adjacent to~$T_{i_k, j}$; if~$\ell_{i_k} = \neg x_{i_k}$ is a negative occurrence, make~$a_{j,k}$ adjacent to~$F_{i_k, j}$ instead.
\end{enumerate}

This completes the construction of the list-coloring instance~$(G_1, L_1)$.

\begin{claim} \label{claim:sat:iff:listcolor}
Formula~$\phi$ is satisfiable~$\Leftrightarrow G_1$ has a proper $3$-coloring respecting the lists~$L_1$.
\end{claim}

\begin{claimproof}
($\Rightarrow$) Suppose that~$\phi$ has a satisfying assignment~$v \colon [n] \to \{\true, \false\}$. Construct a coloring~$f$ of~$G_1$ as follows.
\begin{enumerate}
	\item For each~$i \in [n]$ such that~$v(i) = \true$, set~$f(T_{i,j}) = 2$ and~$f(F_{i,j}) = 3$ for all~$j \in [m]$.
	\item For each~$i \in [n]$ such that~$v(i) = \false$, set~$f(T_{i,j}) = 3$ and~$f(F_{i,j}) = 2$ for all~$j \in [m]$.
\end{enumerate}
It remains to extend the coloring to the clause paths. Since there are no edges between clause paths, this can be done independently for all such paths. So let~$j \in [m]$ and consider the clause path~$D_j$. Each $a$-vertex on~$D_j$ is adjacent to exactly one vertex of a variable path, which represents a literal of~$C_j$ and to which we already assigned a color; we call this vertex the \emph{special neighbor} of the $a$-vertex. For each~$k \in [|C_j|]$ for which the special neighbor of~$a_{j,k}$ is colored~$2$, set~$f(a_{j,k}) = 3$. Since the $a$-vertices are pairwise nonadjacent, this does not introduce any conflicts. If the~$k$'th literal of clause~$C_j$ evaluates to \true under~$v$, then the corresponding vertex is colored~$2$ by the process above, and~$a_{j,k}$ will be colored~$3$. Since each clause contains at least one literal that evaluates to \true, we assign at least one $a$-vertex of~$D_j$ color~$3$ in this way. Color the remaining vertices of~$D_j$ as follows. 
\begin{enumerate}
	\item Consider the prefix~$D'_j$ of~$D_j$ starting from~$a_{j,1}$ up to (but not including) the first $a$-vertex that is colored~$3$. (If we assigned~$a_{j,1}$ the color~$3$ then this prefix is empty and we skip this step.) Assign the $a$-vertices in the prefix the color~$2$ and the $b$-vertices the color~$1$.
	\item Assign the remaining uncolored~$a$-vertices the color~$1$, and assign the remaining uncolored~$b$-vertices the color~$2$.
\end{enumerate}
It is straightforward to verify that each vertex on~$D_j$ is assigned a color from its list, and that the colors of adjacent vertices on~$D_j$ are distinct. To see that no vertex of~$D_j$ is assigned the same color as a neighbor on a variable path, it suffices to observe the following. For $a$-vertices whose special neighbor represents a literal evaluating to \true under~$v$, the literal is colored~$2$ and the $a$-vertex is colored~$3$. For $a$-vertices whose adjacent literal evaluates to \false, the literal is colored~$3$ and the $a$-vertex is colored~$2$ (if it belongs to the prefix~$D'_j$) or~$1$ (if it does not). After extending the partial coloring to each clause path independently, we obtain a proper list coloring of~$G_1$.

($\Leftarrow$) For the reverse direction, suppose that~$G_1$ has a proper list coloring~$f \colon V(G_1) \to [3]$. Create a truth assignment~$v \colon [n] \to \{\true, \false\}$ by setting~$v(i) = \true$ if all $a$-vertices on the variable path~$X_i$ are colored~$2$, and setting~$v(i) = \false$ if the $a$-vertices are all colored~$3$. To see that this is a satisfying assignment, consider an arbitrary clause~$C_j$. As remarked during the construction of~$G_1$, the fact that the clause path~$D_j$ is properly list colored implies that at least one vertex of~$D_j$ has color~$3$. Since~$3$ does not appear on the list of the $b$-vertices, some~$a$-vertex~$a_{j,k}$ is colored~$3$. This implies that the neighbor of~$a_{j,k}$ on a variable path is colored~$2$, implying that a vertex representing the $k$'th literal of~$C_j$ was colored~$2$. Hence the $k$'th literal of~$C_j$ evaluates to \true under the constructed assignment, implying that~$v$ satisfies all clauses.
\end{claimproof}

The cutwidth of~$G_1$ is~$n + \Oh(1)$, but we will not prove this fact. Instead, we will bound the cutwidth of the graphs~$G_2$ and~$G_3$ constructed in next phases.

\subparagraph{(II) Reducing to plain $3$-coloring}

Next we reduce~$(G_1, L_1)$ to an equivalent \textsc{$3$-Coloring} instance~$G_2$ while controlling the cutwidth. In most scenarios, reducing from \textsc{List $q$-Coloring} to plain \textsc{$q$-Coloring} is easy: insert a clique of vertices~$\{p_1, \ldots, p_q\}$ into the graph, and for each color~$i$ that does \emph{not} appear on the list of a vertex~$v$, make~$v$ adjacent to~$p_i$. We cannot use this approach here, since it can cause some~$p_i$ to have degree proportional to the size of the graph. As the size of the graph~$G_1$ is linear in the number of clauses~$m$ of the $n$-variable formula~$\phi$, and the cutwidth of a graph with a vertex of degree~$d$ is at least~$\lfloor d/2 \rfloor$, this would not result in a graph of cutwidth~$n + \Oh(1)$. Hence we are forced to take a more involved approach. Rather than having a single vertex~$p_i$ that is used to block color~$i$ from the lists of all vertices on which~$i$ may not be used, we create a chain of triangles in which the color~$i$ repeats. Then we connect each repetition to only constantly many vertices, to avoid creating high-degree vertices. 

For an integer~$t$, by \emph{a chain of~$t$ triangles} we mean the structure consisting of~$t$ vertex-disjoint triangles~$Z_1, \ldots, Z_t$, in which each triangle~$Z_i$ contains a distinguished \emph{terminal vertex~$z_i$}, such that for each~$i \in [t-1]$ the terminal vertex~$z_{i+1}$ is adjacent to the two non-terminal vertices in~$Z_i$. 

\begin{observation} \label{obs:chain:repeats:color}
In a proper $3$-coloring of a chain of triangles, all terminal vertices receive the same color.
\end{observation}

\begin{figure}[t]
\centering
\includegraphics{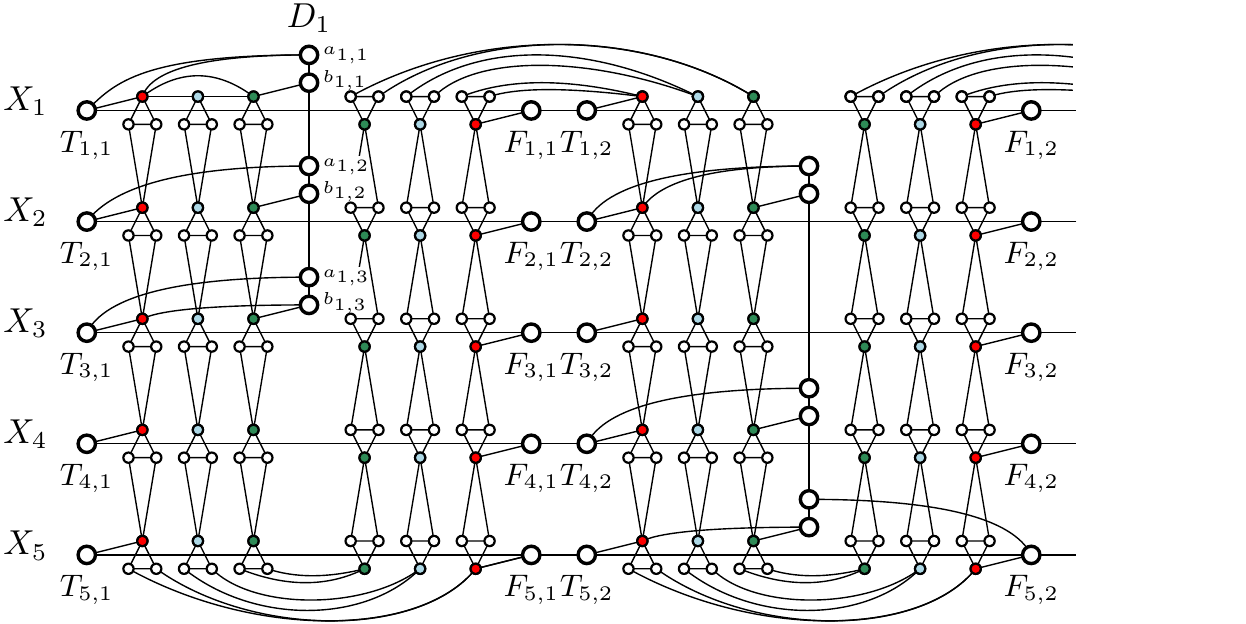}
\caption{Part of the \textsc{$3$-Coloring} instance~$G_2$ that is constructed from the \textsc{List-$3$-Coloring} instance~$(G_1,L_1)$ from Figure~\ref{fig:listthreecol:cw:lb}. The terminals in the chains~$\mathcal{Z}_1, \mathcal{Z}_2, \mathcal{Z}_3$ have been colored red, blue, and green, respectively. The triangles in the chains are numbered from~$1$ to~$2nm$, with the first triangle of each chain shown in the top left. Note that the terminals~$z_{1,1}, z_{2,1}, z_{3,1}$ form a triangle.}
\label{fig:threecol:cw:lb}
\end{figure}

Using this notion we transform~$(G_1,L_1)$ into a plain \textsc{$3$-Coloring} instance~$G_2$, as follows (see Figure~\ref{fig:threecol:cw:lb}).
\begin{enumerate}
	\item Initialize~$G_2$ as a copy of~$G_1$.
	\item For each color~$c \in [3]$, create a chain~$\mathcal{Z}_c$ of~$2nm$ triangles. Refer to the triangles in chain~$\mathcal{Z}_c$ as~$Z_{c,1}, \ldots, Z_{c, 2nm}$ and to the terminal vertices as~$z_{c,1}, \ldots, z_{c,2nm}$. \label{step:make:triangle}
	\item Insert edges to turn~$\{z_{1,1}, z_{2,1}, z_{3,1}\}$ into a triangle. This ensures that among these three vertices, each of the three colors appears exactly once in a proper $3$-coloring. Next, we connect vertices of~$V(G_1)$ to terminal vertices of the chains~$\mathcal{Z}_1, \mathcal{Z}_2, \mathcal{Z}_3$ to enforce the lists. 
	\item For each~$i \in [n]$, for each~$j \in [m]$, for each color~$c \in [3]$, do the following.
	\begin{itemize}
		\item If~$c \notin L_1(T_{i,j})$, then make~$T_{i,j}$ adjacent to~$z_{c, i + (j-1)2n}$.
		\item If~$c \notin L_1(F_{i,j})$, then make~$F_{i,j}$ adjacent to~$z_{c, (2n+1-i) + (j-1)2n}$.
	\end{itemize}
	\item For each~$j \in [m]$, for each~$k \in [|C_j|]$, for each color~$c \in [3]$, do the following.
	\begin{itemize}
		\item If~$c \notin L_1(a_{j,k})$, then make~$a_{j,k}$ adjacent to~$z_{c, i + (j-1)2n}$.
		\item If~$c \notin L_1(b_{j,k})$, then make~$b_{j,k}$ adjacent to~$z_{c, i + (j-1)2n}$.
	\end{itemize}
\end{enumerate}

This concludes the description of the graph~$G_2$. The way in which vertices of~$V(G_1)$ are connected to vertices on the chains will be exploited later, when planarizing the graph.

\begin{claim} \label{claim:listcolor:iff:threecolor}
Graph~$G_2$ has a proper $3$-coloring~$\Leftrightarrow$ graph~$G_1$ has a proper $3$-coloring respecting the lists~$L_1$.
\end{claim}
\begin{claimproof}
($\Rightarrow$) Suppose~$G_2$ has a proper $3$-coloring~$f_2$. Since~$\{z_{1,1}, z_{2,1}, z_{3,1}\}$ forms a triangle by Step~\ref{step:make:triangle}, each of these vertices has a unique color. Permute the color set of~$f_2$ so that~$z_{c,1}$ receives color~$c$, for all~$c \in [3]$. By Observation~\ref{obs:chain:repeats:color} this implies that \emph{all} terminals on chain~$\mathcal{Z}_c$ receive color~$c$. For each vertex~$v \in V(G_1) \cap V(G_2)$, for each color~$c \in [3] \setminus L_1(v)$, we made~$v$ adjacent to a terminal vertex of chain~$\mathcal{Z}_c$, implying that~$v$ does not receive color~$c$ under~$f_2$. Hence each vertex of~$V(G_1)$ receives a color from its list. Since~$G_1$ is a subgraph of~$G_2$, it follows that~$f_2$ restricted to the vertices in~$G_1$ forms a proper list coloring of~$G_1$.

($\Leftarrow$) Suppose~$G_1$ has a proper $3$-coloring~$f_1$ that respects the lists~$L_1$. Create a coloring~$f_2$ of~$G_2$ as follows. For each vertex~$v \in V(G_1) \cap V(G_2)$, set~$f_2(v) := f_1(v)$. For~$c \in [3]$, give all terminals of chain~$\mathcal{Z}_c$ the color~$c$, and give the non-terminal vertices in each triangle distinct colors unequal to~$c$. It is easy to verify that this results in a proper $3$-coloring of~$G_2$.
\end{claimproof}

We construct a drawing~$\psi$ of~$G_2$ as follows; see Figure~\ref{fig:threecol:cw:lb}.
\begin{enumerate}
	\item Draw the variable paths horizontally stacked above one another, such that the indices of the vertices along one variable path~$X_i$ increase from left to right, and such that variable path~$X_{i+1}$ is drawn below~$X_i$ for all~$i \in [n-1]$. Ensure that for all~$j \in [m]$, the vertices~$\{T_{i,j} \mid i \in [n]\}$ are on a vertical line, and the vertices~$\{F_{i,j} \mid i \in [n]\}$ are on a vertical line. For a given~$j \in [m]$, we will refer to the area of the plane enclosed by these two vertical lines as \emph{the $j$'th column}. The horizontal lines for the variable paths divide each column into~$n+1$ \emph{cells}: one cell above each variable path~$X_i$ in the column, and one more cell at the bottom of the column. We refer to the cell just above variable path~$X_i$ in the $j$'th column as~$A_{i,j}$; the bottom cell of the column is~$A_{n+1,j}$.
	\item For each~$j \in [m]$, draw clause path~$D_j$ vertically in the $j$'th column such that the indices of its vertices increase from top to bottom. For~$k \in [|C_j|]$, draw the two vertices~$a_{j,k}, b_{j,k}$ just above the horizontal line for the variable path representing the $k$'th literal of~$C_j$.
	\item For each color~$c \in [3]$, the chain of triangles~$\mathcal{Z}_c$ winds up and down in the~$m$ columns of the drawing. Exactly~$2n$ triangles of the chain are drawn in each column. The first triangle of the chain is drawn in the top-left corner of the first column. Each of the~$2n$ triangles in the $j$'th column crosses a variable path, as visualized in Figure~\ref{fig:threecol:cw:lb}. Note that no two triangle-edges cross in the drawing; the drawing of~$G_2[\mathcal{Z}_1 \cup \mathcal{Z}_2 \cup \mathcal{Z}_3]$ is planar, including the triangle constructed in Step~\ref{step:make:triangle}. The chain~$\mathcal{Z}_1$ is on the left when going down in the column and on the right when going up on the other side. The chain~$\mathcal{Z}_3$ is on the right when going down and on the left when going up, and~$\mathcal{Z}_2$ is in between them.
	\item We draw the edges connecting variable paths to triangle chains. Each vertex on a variable path~$X_i$ is adjacent to a terminal vertex of~$\mathcal{Z}_1$, since~$1 \notin L_1(T_{i,j})$ and~$1 \not\in L_1(F_{i,j})$ for any~$i$ and~$j$. Due to the layout of the chains, the edge from a $T$-vertex or~$F$-vertex to the terminal on the~$\mathcal{Z}_1$ chain to which it is adjacent, can be drawn without crossings, in the same cell as the corresponding terminal vertex on~$\mathcal{Z}_1$. There are no other edges between variable paths and triangle chains.
	\item Finally, we draw the edges that connect a clause path to a triangle chain or variable path. The connections from vertices on the clause paths~$D_j$ to their special neighbor on the variable paths, and to terminal vertices on the triangle chains, may cross other edges. We draw them such that they stay within the cell of the drawing in which the vertex from the clause path is drawn, and only cross the connections between successive triangles in that cell of the drawing.
\end{enumerate}

This concludes the description of the drawing~$\psi$. It is straight-forward (but tedious) to automate this process; there is a polynomial-time algorithm constructing a drawing with the described properties.

\begin{observation} \label{obs:drawing:matching:uncrossed}
For any~$j \in [m-1]$ the edges~$\{ \{F_{i,j}, T_{i,j+1}\} \mid i \in [n]\}$ are not crossed in drawing~$\psi$, and the edges that connect triangles in the $j$'th column to triangles in the $j+1$'th column are not crossed either.
\end{observation}

\begin{claim} \label{claim:cutwidth:gtwo}
A linear layout~$\pi_2$ of~$G_2$ with~$\cutw_{\pi_2}(G_2) \leq n + \Oh(1)$ can be constructed in polynomial time.
\end{claim}
\begin{claimproof}
We describe the linear layout based on the subdivision of the drawing~$\psi$ of~$G_2$ into cells. Consider a cell~$A_{i,j}$ ($i \in [n+1], j \in [m]$) of the drawing of~$G_2$, as illustrated in Figure~\ref{fig:threecol:cw:lb}. Associate to this cell the vertices lying inside the cell, together with the two vertices~$T_{i,j}, F_{i,j}$ on the boundary if~$i \leq n$. Then every vertex of~$G_2$ belongs to exactly one cell, each cell contains constantly many vertices, and the neighbors of a vertex in cell~$A_{i,j}$ lie in the same cell, or one of the four adjacent cells. There is a constant number of edges connecting cell~$A_{i,j}$ to~$A_{i+1,j}$ for any choice of~$i$ and~$j$. Crucially, for~$1 < i \leq n$ there is only a single edge connecting cell~$A_{i,j}$ to~$A_{i,j+1}$: this is the edge~$\{F_{i,j}, T_{i,j+1}\}$. (For~$i=1$ there are six edges on triangle chains connecting~$A_{1,j}$ to~$A_{1,j+1}$, and for~$i=n+1$ there are no edges to~$A_{n+1,j+1}$.)

Consider a linear layout~$\pi_2$ of~$G_2$ which enumerates the cells in column-major order: it enumerates all cells of column~$j$ before moving on to column~$j+1$, and within one column it enumerates all vertices of a cell~$A_{i,j}$ before moving on to cell~$A_{i+1,j}$. The relative order of vertices from the same cell is not important. 

We claim that~$\cutw_{\pi_2}(G_2) \leq n + \Oh(1)$. To see this, consider an arbitrary position~$p$ in the ordering and consider the cut~$\mathcal{C}$ of edges that connect a vertex with index at most~$p$, to a vertex with index greater than~$p$. Let~$i \in [n+1]$ and~$j \in [m]$ such that vertex~$\pi_2(p)$ belongs to cell~$A_{i,j}$. We now bound the number of edges in~$\mathcal{C}$ based on the cells that contain their endpoints. Since~$\pi_2$ enumerated cells in column-major order, and vertices in a cell are only adjacent to vertices in adjacent cells, the edges in cut~$\mathcal{C}$ must be of one of the following types:
\begin{itemize}
	\item Edges that connect two vertices from~$A_{i,j}$ to each other, or that connect~$A_{i,j}$ to~$A_{i,j-1} \cup A_{i,j+1}$. There are~$\Oh(1)$ such edges because a cell contains~$\Oh(1)$ vertices.
	\item Edges that connect a vertex from~$A_{i,j}$ to a vertex in~$A_{i-1,j} \cup A_{i+1,j}$. Again there are~$\Oh(1)$ such edges.
	\item Edges that connect a cell~$A_{i', j}$ for~$i' < i$ to the cell on its right. Since horizontally-neighboring cells are connected by exactly one edge for~$1<i'\leq n$, and by~$\Oh(1)$ edges otherwise, there are~$i + \Oh(1)$ such edges.
	\item Edges that connect a cell~$A_{i', j-1}$ for~$i' > i$ to the cell on their right. There are at most~$n-i$ such edges.
\end{itemize} 
This accounts for all possible edges in~$\mathcal{C}$. In particular, that there are no edges in~$\mathcal{C}$ whose endpoints both lie in cells that come before~$A_{i,j}$ in the column-major order, or both lie in cells coming after~$A_{i,j}$. It follows that cut~$\mathcal{C}$ contains at most~$n + \Oh(1)$ edges, which proves that~$\cutw_{\pi_2}(G_2) \leq n + \Oh(1)$. It is trivial to construct~$\pi_2$ in polynomial time.
\end{claimproof}

\subparagraph{(III) Reducing to a planar graph} The final step of the construction turns~$G_2$ into a planar graph of cutwidth~$n + \Oh(1)$. Let~$G_3$ be obtained by applying Theorem~\ref{thm:planarize:threecol} on the drawing~$\psi$ of~$G_2$ constructed above. The choice of which endpoint to distinguish in Step~\ref{step:threecol:merge} of the transformation can be made arbitrarily. See Figure~\ref{fig:threecol:complete:lb} for an illustration. By Theorem~\ref{thm:planarize:threecol}, we know that~$G_3$ is $3$-colorable if and only if~$G_2$ is; by Claims~\ref{claim:sat:iff:listcolor} and~\ref{claim:listcolor:iff:threecolor}, this happens if and only if~$\phi$ is satisfiable. It is straightforward to implement the transformation in polynomial time. The following claim bounds the cutwidth of~$G_3$.

\begin{figure}[t]
\centering
\includegraphics{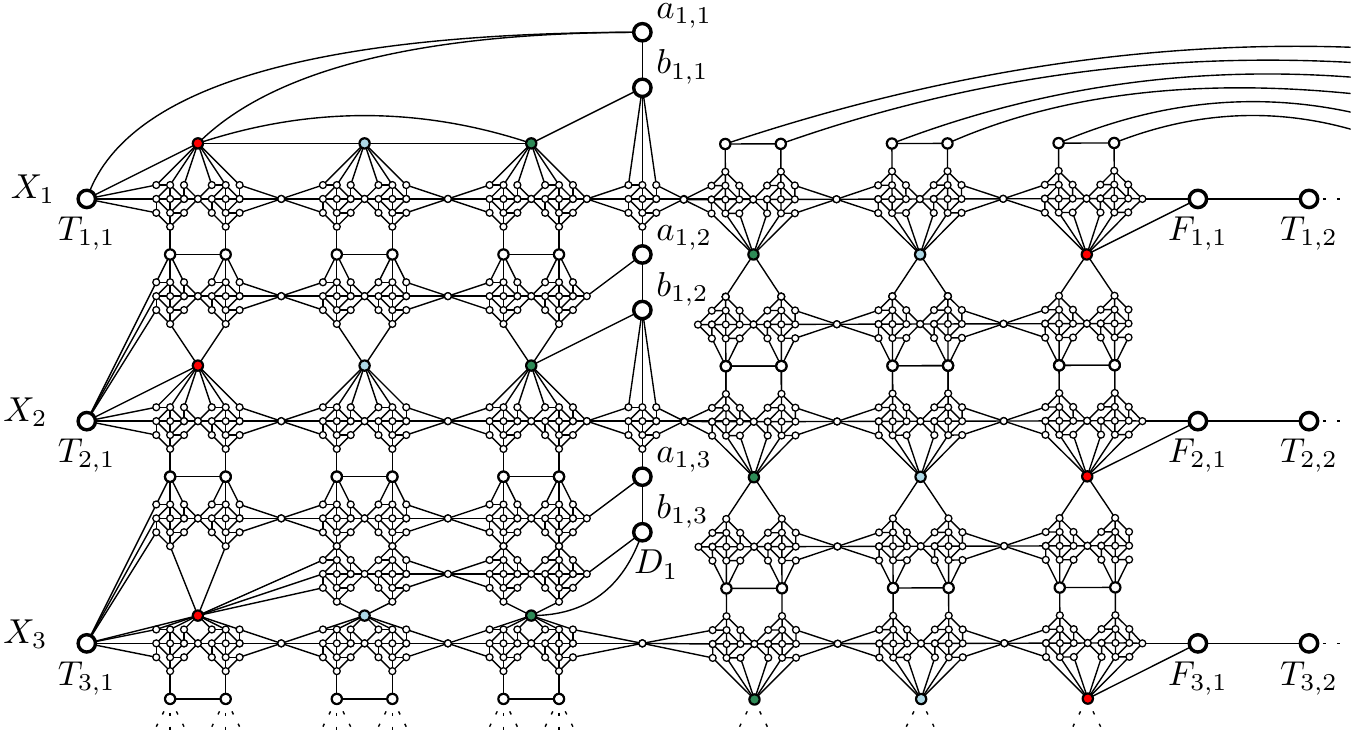}
\caption{Illustration of the planar \textsc{$3$-Coloring} instance~$G_3$ that is constructed from the \textsc{$3$-Coloring} instance~$G_2$ from Figure~\ref{fig:threecol:cw:lb}. Only the top-left part of the instance is drawn.}
\label{fig:threecol:complete:lb}
\end{figure}

\begin{claim}
A linear layout~$\pi_3$ of~$G_3$ with~$\cutw_{\pi_3}(G_3) \leq n + \Oh(1)$ can be constructed in polynomial time.
\end{claim}
\begin{claimproof}
The strategy used in the proof of Claim~\ref{claim:cutwidth:gtwo} can be easily adapted for~$G_3$. The drawing of~$G_3$ that is obtained from Theorem~\ref{thm:planarize:threecol} can be used to partition~$V(G_3)$ into columns and cells, such that ordering the vertices in column-major order achieves the desired cutwidth. Concretely, the planarization procedure inserts copies of~$\Hcol$ into cells of the drawing of~$G_2$. Each such inserted vertex is associated to the cell in which the inserted vertex lies, with ties broken arbitrarily. By Observation~\ref{obs:drawing:matching:uncrossed}, the single edge connecting a cell~$A_{i,j}$ for~$1<i\leq n$ to the cell~$A_{i,j+1}$ in~$G_2$ is not crossed, and therefore remains a single connecting edge in~$G_3$. As the number of vertices in each cell increases by a constant, the argumentation of Claim~\ref{claim:cutwidth:gtwo} goes through unchanged to argue that a column-major ordering~$\pi_3$ of~$V(G_3)$ has cutwidth~$n + \Oh(1)$. It can easily be constructed in polynomial time.
\end{claimproof}

This completes the proof of Lemma~\ref{lemma:coloring:lb:construction}.
\end{proof}

Using Lemma~\ref{lemma:coloring:lb:construction}, it is easy to prove the claimed runtime lower bound for \textsc{$3$-Coloring}.

\begin{proof}[Proof of Theorem~\ref{thm:3collow}]
Suppose \textsc{$3$-Coloring} on a planar graph with a given linear layout of cutwidth~$k$ can be solved in~$\Oh^*((2-\varepsilon)^k)$ time for some~$\varepsilon > 0$, by an algorithm called~$A$. Then \textsc{CNF-SAT} with clauses of arbitrary size can be solved in~$\Oh^*((2-\varepsilon)^k)$ time by turning an input formula~$\phi$ into a planar graph~$G$ with linear layout~$\pi$ of cutwidth~$n + \Oh(1)$ using Lemma~\ref{lemma:coloring:lb:construction}, and then running~$A$ on~$(G,\pi)$. This contradicts SETH.
\end{proof}

\subsection{Lower bounds for coloring on graphs of bounded pathwidth and degree}
In this section we prove that the base of the exponent in the running time of Theorem~\ref{thm:alg:degree} is optimal for every odd~$d \geq 5$. Our proof is inspired by a construction due to Jaffke and Jansen~\cite[Theorem 15]{JaffkeJ17} which gives a lower bound for solving \qColoring parameterized by the vertex-deletion distance to a linear forest. Their result, in turn, extends the original lower bound of Lokshtanov et al.~\cite{LokshtanovMS11} for \qColoring on graphs of bounded pathwidth or feedback vertex number. We shall use the following lemma to construct \textsc{List Coloring} gadgets. (We will eliminate the need for having lists later in the construction.)

\begin{lemma}[{\cite[Lemma 14]{JaffkeJ17}}]\label{lem:coloring:pathgadget}
For each~$q \geq 3$ there is a polynomial-time algorithm that, given $(c_1, \ldots, c_m) \in [q]^m$, outputs a $q$-list-coloring instance~$(P, \Lambda)$ where~$P$ is a path of size~$\Oh(m)$ with distinguished vertices~$(\pi_1, \ldots, \pi_m)$, such that the following holds. For each~$(d_1, \ldots, d_m) \in [q]^m$ there is a proper list-coloring~$\gamma$ of~$P$ in which~$\gamma(\pi_i) \neq d_i$ for all~$i$, if and only if~$(c_1, \ldots, c_m) \neq (d_1, \ldots, d_m)$.
\end{lemma}

The lemma gives a way to construct a small path that forbids a specific coloring to be used on a given set of vertices, in a list-coloring instance. If~$(v_1, \ldots, v_m)$ are vertices in a graph under construction for which we want to forbid a certain coloring~$(c_1, \ldots, c_m) \in [q]^m$ (i.e.~we want to ensure that it is not the case that~$v_i$ will be colored~$c_i$ for all~$i \in [m]$), then this can be achieved as follows. We create a gadget~$(P, \Lambda)$, add it to the graph, and connect each vertex~$v_i$ to the corresponding distinguished vertex~$\pi_i$ on the path. Given a coloring on the vertices~$v_1, \ldots, v_m$, if we want to find a proper coloring of the path~$P$, we have to assign each~$\pi_i$ a color different from its partner~$v_i$; hence~$\pi_i$ has to avoid the color~$d_i$ given to~$v_i$. The lemma guarantees that each~$\pi_i$ can avoid its forbidden color~$d_i$, if and only if~$(d_1, \ldots, d_m) \neq (c_1, \ldots, c_m)$; hence a coloring of~$(v_1, \ldots, v_m)$ can be extended to~$P$ if and only if it is not the coloring~$(c_1, \ldots, c_m)$ the gadget is meant to forbid.

The second gadget we need in the construction is a chain of cliques, to propagate a coloring choice throughout a small-pathwidth graph. For integers~$t \in \mathbb{N}$ and~$q \in \mathbb{N}$, define a \emph{$t$-chain of $q$-cliques} as the graph constructed as follows. Start from a disjoint union of~$t$ cliques~$Z_1, \ldots, Z_t$ of size~$q$ each, in which each clique~$Z_i$ contains a distinguished \emph{terminal vertex}~$z_i \in Z_i$. For	each~$i \in [t-1]$, connect the terminal vertex~$z_{i+1} \in Z_{i+1}$ to the~$q-1$ non-terminal vertices in~$Z_i$.

\begin{proposition}
In any proper $q$-coloring of a $t$-chain of~$q$-cliques, all distinguished vertices have the same color.
\end{proposition}
\begin{proof}
In the first $q$-clique~$Z_1$, all~$q$ colors occur exactly once. Since~$z_2$ is adjacent to all vertices in~$Z_1$ except the terminal, the $(q-1)$-colors of the non-terminals in~$Z_1$ cannot be used on terminal~$z_2$, hence~$z_2$ must receive the same color as~$z_1$. Repeating this argument shows that all terminals have the same color.
\end{proof}

Observe that the maximum degree in a $t$-chain of $q$-cliques is~$2(q-1)$, which is achieved by any terminal vertex~$z_i$ in the interior of the chain. It is adjacent to~$q-1$ vertices in its own clique~$Z_i$, and to $q-1$ non-terminal vertices in~$Z_{i-1}$. 

The main idea of the reduction is that by creating a chain of cliques, we can propagate a choice with~$q$ possibilities (the color of the terminals) throughout a path decomposition, using vertices of degree at most~$2(q-1)$. We will encode truth assignments to variables of a \textsc{CNF-SAT} instance through colors given to such chains. We will enforce that the encoded truth assignment satisfies a clause, by enforcing that an assignment that does \emph{not} satisfy the clause, is not the one encoded by the coloring. To check this, we take one terminal from each chain and connect it to a partner on a path gadget~$(P, \Lambda)$ that forbids a specific coloring. Hence each vertex on a chain will receive at most one more neighbor, giving a maximum degree of~$d := 2(q-1)+1 = 2q-1$ to represent a \qColoring instance. Then solving this \qColoring instance in~$\Oh^*((\lfloor d/2 \rfloor + 1 - \varepsilon)^\pw) = \Oh^*(((q-1) + 1 - \varepsilon)^\pw)$ time will contradict SETH for the same reason as in the earlier construction showing the impossibility of~$\Oh^*((q-\varepsilon)^\pw)$-time algorithms.

We are now ready to prove the main result of this section. Note that Theorem~\ref{thm:degreecollow} is a direct consequence of this result.
\begin{theorem}
	Let~$d \geq 5$ be an odd integer and let~$q_d := \lfloor d /2 \rfloor + 1$. Assuming SETH, there is no~$\varepsilon > 0$ such that \textsc{$q_d$-Coloring} on a graph of maximum degree~$d$ given along with a path decomposition of pathwidth~$k$ can be solved in time~$\Oh^*((\lfloor d/2 \rfloor + 1 - \varepsilon)^k)$.
\end{theorem}
\begin{proof}
Suppose there is some~$\varepsilon > 0$ such that \qdColoring can be solved in time~$\Oh^*((\lfloor d/2 \rfloor + 1 - \varepsilon)^k) = \Oh^*((q_d - \varepsilon)^k)$ time. We will show that this implies the existence of~$\delta > 0$ such that for each constant~$s \in \mathbb{N}$, the \textsc{CNF-SAT} problem with clauses of size at most~$s$ (\sSAT) can be solved in~$\Oh^*((2 - \delta)^n)$ time for formulas with~$n$ variables. (We refer to this problem as \sSAT instead of \qSAT in this proof, to avoid confusion with the number of colors.)

We choose an integer~$p$ depending on~$\varepsilon$, in a way explained later. Take an $n$-variable input formula~$\phi$ of \sSAT for some constant~$s$. Let~$C_1, \ldots, C_m$ be the clauses of~$\phi$, and assume without loss of generality that no clause contains contradicting literals and that each variable occurs in at least one clause. Split the variables~$x_1, \ldots, x_n$ of~$\phi$ into groups~$F_1, \ldots, F_t$, each of size at most~$\beta := \lfloor \log ((q_d)^p) \rfloor = \lfloor p \log q_d \rfloor$ so that~$t = \lceil n/\beta \rceil$. A truth assignment to the variables in one group~$F_i$ is a \emph{group assignment}. A group assignment for~$F_i$ satisfies a clause~$C_j$ if~$F_i$ contains a variable~$x_\ell$ such that (i)~$x_\ell$ is set to \true by the group assignment and~$x_\ell$ is a literal in~$C_j$, or (ii)~$x_\ell$ is set to \false by the group assignment and~$\neg x_\ell$ is a literal in~$C_j$.

Let~$M := m \cdot (q_d)^{p \cdot s}$. For each group index~$i \in [t]$, we construct $p$ distinct $M$-chains of~$q_d$-cliques called~$\mathcal{Z}_i^1, \ldots, Z_i^p$. For each chain~$\mathcal{Z}_i^j$, let~$z_i^{j, 1}, \ldots, z_i^{j, M}$ be its terminal vertices. For each group index~$i \in [t]$, let~$\mathcal{V}_i := \{z_i^{1,1}, \ldots, z_{i}^{p, 1}\}$ denote the~$p$ first vertices on each chain~$\mathcal{Z}_i^1, \ldots, \mathcal{Z}_i^p$. Since there are~$(q_d)^p$ distinct $q_d$-colorings of~$\mathcal{V}_i$, and~$F_i$ contains at most~$\beta = \lfloor \log ((q_d)^p) \rfloor$ variables, it follows that the number of $q_d$-colorings of~$\mathcal{V}_i$ is at least as large as the number of group assignments for~$F_i$ which is~$2^{|F_i|} \leq 2^\beta \leq (q_d)^p$. For each group~$i \in [t]$, we can therefore find an efficiently computable injection~$g_i \colon 2^{|F_i|} \to [q_d]^p$ that assigns to each group assignment for~$F_i$ a unique $(q_d)$-coloring of~$\mathcal{V}_i$.

The colorings of the chains constructed for the groups~$F_i$ will represent a truth assignment for the variables of~$\phi$. In addition, we create~$q_d$ more chains that will be used to simulate list-coloring behavior of the path gadgets constructed by Lemma~\ref{lem:coloring:pathgadget}. Let~$N \in \Oh((q_d)^s)$ such that when applying Lemma~\ref{lem:coloring:pathgadget} to obtain a path gadget to block a coloring of at most~$(q_{d})^s$ vertices, the resulting path gadget has at most~$N$ vertices. We create a \emph{color chain}~$\mathcal{Y}_i$ for each~$i \in [q_d]$, which is an $NM$-chain of~$q_d$-cliques~$Y_i^1, \ldots, Y_i^{NM}$ with terminal vertices~$y_i^1, \ldots, y_i^{NM}$. We connect the first terminal vertex on each color chain to the first terminal vertex of the other color chains, so that they form a clique of size~$q_d$ called the \emph{palette}. In a proper $q_d$-coloring, each vertex of the palette receives a unique color, and this color repeats on all later terminal vertices on the same chain. To simulate a vertex~$v$ in a list-coloring gadget whose list forbids the use of color~$i$, we will connect~$v$ to a terminal on color chain~$\mathcal{Y}_i$.

The chains constructed so far from the heart of the instance. The remainder of the instance will be formed by path gadgets constructed using Lemma~\ref{lem:coloring:pathgadget}, which will enforce that a proper $q_d$-coloring of the graph encodes a truth assignment that satisfies all clauses. 

For each clause index~$j \in [m]$, we augment the graph to enforce that~$C_j$ is satisfied. The clause~$C_j$ contains~$s' \leq s$ variables $x_{j'_1}, \ldots, x_{j'_{s'}}$, which are spread over~$s'' \leq s'$ different groups $F_{j_1}, \ldots, F_{j_{s''}}$ whose group assignments are encoded by colorings of~$\mathcal{V}_{j_1}, \ldots, \mathcal{V}_{j_{s''}}$. There is exactly one assignment to the variables in~$C_j$ that does \emph{not} satisfy clause~$C_j$. Now do the following. Consider the tuples of colorings~$(f_1 \colon \mathcal{V}_{j_1} \to [q_d]^p, \ldots, f_{s''} \colon \mathcal{V}_{j_{s''}} \to [q_d]^p)$ with the property that (i)~some coloring in this tuple does not correspond to a group assignment, or (ii)~all colorings correspond to group assignments, but none of these group assignments satisfy clause~$C_j$. For each such tuple, construct a $q_d$-list-coloring instance~$(P, \Lambda)$ using Lemma~\ref{lem:coloring:pathgadget} for the forbidden coloring given by (the concatenation of)~$f_1, \ldots, f_{s''}$. To forbid this coloring from appearing in a solution, we could connect the distinguished vertices on~$P$ to the corresponding vertices in the sets~$\mathcal{V}_{j_1}, \ldots, \mathcal{V}_{j_{s''}}$. However, that would blow up the degree of those vertices. Instead, we will use the fact that the chains of cliques force the coloring to repeat, so that instead of connecting the distinguished vertices on~$P$ to~$\mathcal{V}_{j_1}, \ldots, \mathcal{V}_{j_{s''}}$, we may equivalently connect them to later terminal vertices on the same chains. When considering the $r$'th tuple of colorings that fail to satisfy clause~$C_j$, we insert the path~$P$ into the graph and connect its distinguished vertices to the terminals with index~$(j-1) \cdot (q_d)^{p \cdot s} + r$ of the relevant chains. Observe that~$M$ is chosen such that for each clause index, we have fresh terminals to accommodate all of the at most~$(q_d)^{p \cdot s}$ ways in which it can fail to be satisfied. Now we will deal with the fact that~$(P, \Lambda)$ is a list-coloring instance whereas we are embedding it in a normal coloring instance. Let~$p_1, \ldots, p_{|P|}$ be (all) the vertices of the path gadget~$P$ in their natural order along the path, and observe that our choice of~$N$ guarantees that~$|P| \leq N$. For the path gadget~$P$ created for the $r$'th tuple of bad colorings for clause~$C_j$, for each vertex~$p_\ell$ on~$P$, do the following. For each color~$c \in [q_d] \setminus \Lambda(p_\ell)$, i.e.~for each color~$c$ that is not allowed to be used on~$p_\ell$ in the list-coloring instance~$(P, \Lambda)$, connect vertex~$p_\ell$ to the terminal vertex with index~$(j-1) \cdot (q_d)^{p \cdot s} \cdot N + N \cdot (r - 1) + \ell$ on color chain~$\mathcal{Y}_c$. 

This concludes the construction of the graph~$G$.

\begin{claim}
The maximum degree of~$G$ is~$d$.
\end{claim}
\begin{claimproof}
We first consider vertices on chains of $q_d$-cliques. Each such vertex has at most~$2(q_d-1)$ neighbors on the chain. In addition, each vertex on a chain is connected to at most one vertex of a gadget path~$P$ over the course of the entire construction, so its degree is at most~$2(q_d - 1) + 1 = d$.

It remains to bound the degree of vertices on a gadget path~$P$. Such a vertex has at most two neighbors on the path, at most~$q_d - 1$ neighbors on color chains corresponding to at most~$q_d - 1$ colors that are not on its list of allowed colors (no color list is empty), and at most one neighbor on a variable-encoding chain. Hence vertices on gadget paths have degree at most~$q_d + 2 \leq d$.
\end{claimproof}

\begin{claim} \label{claim:deg:lb:sound}
$G$ has a proper $q_d$-coloring if and only if~$\phi$ is satisfiable.
\end{claim}
\begin{claimproof}
($\Rightarrow$) Suppose~$\phi$ is satisfied by some truth assignment~$v \colon [n] \to \{\true,\false\}$. We construct a proper coloring of~$G$ with~$q_d$ colors. For each color chain~$\mathcal{Y}_i$ for~$i \in [q_d]$, give all terminals color~$c$, and for each clique on the chain give the non-terminals in the clique distinct colors in~$[q_d] \setminus \{i\}$. For each group~$F_i$ of variables for~$i \in [t]$, consider the group assignment to~$F_i$ induced by~$v$. The group assignment corresponds to a unique coloring of~$\mathcal{V}_i$ according to the mapping~$g_i$; color the vertices~$\mathcal{V}_i$ accordingly, and repeat this coloring on the remaining terminal vertices. Color the non-terminals of each clique on a variable-encoding chain by distinct colors different from the terminal. It remains to show that for each inserted path gadget~$(P,\Lambda)$, the coloring can be extended to the vertices of~$P$. Since the connections from~$P$ to the color chains match the list requirements of~$\Lambda$, for this it suffices to obtain a list coloring of~$P$ in which no vertex of~$P$ is assigned the same color as a neighbor on a variable-encoding chain. But since the path gadgets were only inserted to block colorings representing group assignments that do \emph{not} satisfy a clause~$C_j$, whereas the group assignments induced by~$f$ do satisfy all clauses, it follows that each such gadget~$(P, \Lambda)$ can be list colored while avoiding the colors of its neighbors on variable paths. Performing this extension separately for all inserted path gadgets yields a proper $q_d$-coloring of~$G$.

($\Leftarrow$) Suppose~$G$ has a proper $q_d$-coloring. Since the palette is a clique of size~$q_d$, it contains each color exactly once. By permuting the color set if needed, we obtain a proper $q_d$-coloring~$f \colon V(G) \to [q_d]$ such that~$y_i$ has color~$i$ for all~$i \in [q_d]$. Now consider the coloring of the sets~$\mathcal{V}_i$ that represent the different variable groups~$F_i$. We claim that the coloring of each~$\mathcal{V}_i$ corresponds to a proper group assignment under~$g_i$. Suppose not; consider a variable~$x_\ell \in F_i$ and a clause~$C_j$ involving~$x_\ell$. Then for all tuples of colorings we considered for the variable groups related to clause~$C_j$ in which the coloring of~$F_i$ matches that of~$f$, we inserted a path gadget~$(P, \Lambda)$ to block the coloring. Since the colors of the palette propagate through the color chains, each path gadget is properly list colored. Since terminal vertices on a variable-encoding chain have the same color as the first vertices~$\mathcal{V}_i$ of the chains, by Lemma~\ref{lem:coloring:pathgadget} it follows that the path gadget cannot be properly list colored while avoiding all colors of its neighbors on variable chains; hence the coloring~$f$ is improper.

It follows that the coloring of each~$\mathcal{V}_i$ corresponds to a proper group assignment. Combining these group assignments into a truth assignment~$v \colon [n] \to \{\true, \false\}$ of the entire formula, we obtain a satisfying assignment for~$\phi$. All clauses are satisfied because any tuple of colorings representing group assignments that fail to satisfy the clause, is blocked by a path gadget.
\end{claimproof}

\begin{claim}
A path decomposition of~$G$ of width~$p t + q_d + \Oh(q_d)^s$ can be constructed in polynomial time.
\end{claim}
\begin{claimproof}
To bound the pathwidth we give a vertex search strategy~\cite{EllisST94,KirousisP86}, which can be used to upper bound the pathwidth. This is a strategy for the following game. The edges of the graph are interpreted as a network of tunnels, which are initially filled with contagious gas. One by one, cleaners can be placed or removed from vertices. An edge is \emph{cleared} when both endpoints are occupied by cleaners. An edge is recontaminated if there is ever a path in the graph from a contaminated to an uncontaminated edge, such that no vertex along this path is occupied by a cleaner. The \emph{vertex search number} of the graph is the smallest number of cleaners needed to clear the entire graph; this number is known to equal the pathwidth of~$G$ plus one~\cite[Theorem 4.1]{KirousisP86}. Hence we provide an upper bound on the pathwidth by giving a cleaning strategy with~$pt + q_d + \Oh(q_d)^s$ cleaners.

The cleaning proceeds in rounds which are indexed by a tuple~$(j, r) \in [m] \times [(q_d)^p]$. The rounds follow each other in increasing lexicographical order. A round~$(j,r)$ starts with cleaners on the terminal vertices with index~$(j-1) \cdot (q_d)^{p \cdot s} + r$ on the~$pt$ variable-encoding chains~$\mathcal{Z}$, and cleaners on the terminal vertices with index~$(j-1) \cdot (q_d)^{p \cdot s} \cdot N + N \cdot (r-1) + 1$ on the~$q_d$ color chains~$\mathcal{Y}$, and no other cleaners. The round proceeds by placing cleaners on all vertices of the path~$P$ that was inserted for the $r$'th bad tuple of colorings of clause~$C_j$, if any. One variable-encoding chain at a time, we then (i)~place~$q_d$ additional searchers on the~$q_d - 1$ non-terminal vertices of the current clique on the chain, and on the next terminal vertex~$v$, and afterward~(ii) remove all cleaners except~$v$ from that chain. Similarly, one color-encoding chain~$\mathcal{Y}_c$ at a time, we clean the next~$N$ cliques on~$\mathcal{Y}_c$ by occupying the non-terminals and next terminal, and then removing cleaners from the previous clique. After having effectively moved the cleaners on the color-chains forward by~$N$ positions, and on variable-encoding chains by one position, we remove the cleaners from~$P$ and are in position to start the next round. It is easy to verify that this constitutes a valid cleaning strategy for~$G$. The bottleneck number of simultaneously active cleaners is formed by~$p t + q_d$ (one cleaner per chain) plus~$q_d$ additional cleaners to clean the next clique on a chain, plus~$N$ cleaners on a path gadget~$P$. Since~$q_d + N \in \Oh(q_{d}^s)$ the number of cleaners, and therefore the pathwidth of~$G$, is indeed bounded as required. It is straight-forward to construct a path decomposition of the corresponding width in polynomial time.
\end{claimproof}

Using these claims we finish the proof similarly as in Lemma 6.4 of the work by Lokshtanov et al~\cite{LokshtanovMS11}. We started from the assumption that \qdColoring can be solved in~$\Oh^*((\lfloor d/2 \rfloor + 1 - \varepsilon)^k) = \Oh^*((q_d - \varepsilon)^k)$ time on graphs of maximum degree~$d$, when given a path decomposition of width~$k$. Let~$\lambda := \log_{q_d}(q_d - \varepsilon) < 1$, such that \qdColoring on graphs of maximum degree~$d$ can be solved in time~$\Oh^*(q_d^{\lambda k})$ for graphs of pathwidth~$k$. Choose a sufficiently large integer~$p$ such that~$\delta' = \lambda \frac{p}{p-1}$ is strictly smaller than one. Claim~\ref{claim:deg:lb:sound} shows that an instance~$\phi$ of \sSAT can be solved by constructing the corresponding graph~$G$ and a path decomposition of width~$pt + q_d + c \cdot (q_d)^s$ for some constant~$c$, and solving \qdColoring on this structure. As we keep~$q_d$ fixed,~$s$ is a constant, and~$p$ depends only on~$\varepsilon$ and~$q_d$, the quantity~$q_d^{p \cdot s}$ is a constant. The graph we construct has size polynomial in~$MN$, which is polynomial in the size of the input formula~$\phi$ since~$(q_d)^{p \cdot s}$ is constant. Hence the construction of~$G$ with its path decomposition, and any terms polynomial in~$|G|$, are polynomial in~$|\phi|$. By our choice of~$p$ we have
\[ \lambda p t = \lambda p \left \lceil \frac{n}{\lfloor p \log q_d \rfloor} \right \rceil \leq \lambda p \frac{n}{\lfloor p \log q_d \rfloor} + \lambda p \leq \lambda p \frac{n}{(p-1) \log q_d} + \lambda p \leq \frac{\delta' n}{\log q_d} +\lambda p, \]
for some~$\delta' < 1$. Hence, running the hypothetical algorithm for \qdColoring on~$G$ takes time
\begin{align*}
\Oh^*((q_d)^{\lambda k}) = 
\Oh^*((q_d)^{\lambda p t + \lambda (q_d + c (q_d)^s)}) = 
\Oh^*((q_d)^{\lambda p t}) = \\
\Oh^*((q_d)^{\frac{\delta' n}{\log q_d} + \lambda p}) = 
\Oh^*((q_d)^{\frac{\delta' n}{\log q_d}}) =
\Oh^*(2^{\delta' n}) = 
\Oh^*((2-\delta)^n)
\end{align*}
for some~$\delta > 0$ that does not depend on~$s$. This shows that for each constant~$s$, the \sSAT problem can be solved in time~$\Oh^*((2-\delta)^n)$, contradicting SETH.
\end{proof}

\end{document}